\documentclass[acmsmall, acmthm, natbib=false]{acmart}
\usepackage{microtype}
\usepackage{xspace,amsmath,amsfonts,bbm,ifthen,mathtools}
\usepackage[heavycircles]{stmaryrd}
\usepackage{bbold}
\usepackage[sort, noadjust]{cite}
\usepackage{subcaption}

\usepackage{tikz}
\usetikzlibrary{arrows,automata,shapes,calc,decorations.markings,positioning}
\usetikzlibrary{decorations.pathmorphing}
\tikzset{%
  ->, >=stealth', auto, initial text=,
  state/.style={%
    inner sep=.5mm, minimum size=3.5mm, draw=black, circle
  }
  state with output/.style={%
    shape=rectangle split, rectangle split parts=2, draw, fill=white,
    initial text=, inner sep=1mm
  }
}

\usepackage{algorithm}
\usepackage[indLines=true,noEnd=true,beginComment=//~,beginLComment=/*~, endLComment=~*/]{algpseudocodex}
\tikzset{algpxIndentLine/.style={draw=gray,very thin,-}}
\usepackage{caption}
\usepackage{xfrac}

\theoremstyle{acmdefinition}\newtheorem*{rexample}{Running example}

\newcommand*\ie{\textit{i.e.,}\xspace}
\newcommand*\eg{\textit{e.g.,}\xspace}
\newcommand*\Real{\mathbbm{R}}
\newcommand*\tto[1]{\ensuremath{\xrightarrow{#1}}}
\newcommand*\Int{\mathbbm{Z}}
\newcommand{\first}[1]{\textup{\textsf{first}}(#1)}
\newcommand{\last}[1]{\textup{\textsf{last}}(#1)}
\newcommand*\Nat{\mathbbm{N}}
\newcommand*\weights[3]{\textup{\textsf{weights}}_{#1\downarrow #2}(#3)}
\newcommand*\sweights[2]{\textup{\textsf{weights}}_{#1}(#2)}
\newcommand*\lweight[2]{\textup{\textsf{lastweight}}_{#1}(#2)}

\newcommand*\bigmid{\mathrel{\big|}}
\newcommand*\Realnn{\Real_{\ge 0}}
\renewcommand*\epsilon{\varepsilon}
\renewcommand*\phi{\varphi}
\newcommand*\sem[2][]{\llbracket #2\rrbracket^{#1}}
\newcommand*\oil{\mathopen{]}}
\newcommand*\oir{\mathclose{[}}
\newcommand*\cpa[1]{\textup{\textsf{cpa}}(#1)}
\newcommand*\cpr{\mathfrak{r}}
\newcommand*\bigland{\bigvee}

\begin{document}

\title[$\omega$-Regular Energy Problems]{\texorpdfstring{$\boldsymbol{\omega}$}{omega}-Regular Energy Problems}

\author{Sven Dziadek}
\orcid{0000-0001-6767-7751}
\affiliation{%
  \institution{Inria}
  \city{Paris}
  \country{France}
}

\author{Uli Fahrenberg}
\orcid{0000-0001-9094-7625}
\affiliation{%
  \institution{EPITA Research Laboratory (LRE)}
  \city{Paris}
  \country{France}
}

\author{Philipp Schlehuber-Caissier}
\orcid{0000-0002-6611-9659}
\affiliation{%
  \institution{EPITA Research Laboratory (LRE)}
  \city{Paris}
  \country{France}
}

\acmJournal{FAC}

\begin{abstract}
  We show how to efficiently solve problems involving a quantitative measure,
  here called energy, as well as a qualitative acceptance condition,
  expressed as a Büchi or Parity objective, in finite
  weighted automata and in one-clock weighted timed automata.
  Solving the former problem and extracting the corresponding witness
  is our main contribution and is handled by a
  modified version of the Bellman-Ford algorithm interleaved with Couvreur's
  algorithm.
  The latter problem is handled via a reduction to the
  former relying on the corner-point abstraction.  All our algorithms
  are freely available and implemented in a tool based on the
  open-source platforms TChecker and~Spot.
\end{abstract}

\keywords{weighted timed automaton, weighted automaton, energy problem,
  generalized Büchi acceptance, Parity acceptance, energy constraints}

\maketitle

\section{Introduction}

Energy problems in weighted (timed) automata pose the question whether
there exist infinite runs in which the accumulated weights always stay
positive.
Since their introduction in
\cite{DBLP:conf/formats/BouyerFLMS08}, much research has gone into
different variants of these problems, for example energy games
\cite{DBLP:conf/fsttcs/ChatterjeeDHR10,
  DBLP:conf/ictac/FahrenbergJLS11, DBLP:journals/iandc/VelnerC0HRR15},
energy parity games \cite{DBLP:journals/tcs/ChatterjeeD12}, robust
energy problems \cite{DBLP:journals/fac/BacciBFLMR21}, etc., and into their application in
embedded systems \cite{DBLP:conf/emsoft/FalkHLLP12,
  DBLP:conf/pts/FrehseLMN11}, satellite control
\cite{DBLP:conf/fm/BisgaardGHKNS16,
  DBLP:conf/isola/MikucionisLRNSPPH10}, and other areas.
Nevertheless, many basic questions remain open and implementations are
somewhat lacking.

The above results discuss \emph{looping} automata~\cite{DBLP:conf/focs/WolperVS83},
\ie $\omega$-automata in which all states are accepting.
In practice, looping automata do not suffice
because they cannot express all liveness properties.
For model checking, formal properties (\eg in LTL)
are commonly translated into (generalized) Büchi automata~\cite{buechi_infinite},
or Parity automata \cite{renkin.20.atva} if determinism is of the issue,
which provide models for the class of $\omega$-regular languages.

In this work, we extend energy problems with transition-based
generalized Büchi or Parity conditions and treat them for weighted automata as
well as weighted timed automata with precisely one clock.
On weighted automata we show that they are effectively decidable using a
combination of a modified Bellman-Ford algorithm
\cite{ford, bellman} with Couvreur's algorithm \cite{couvreur99}.
For weighted timed automata we show that one can use the corner-point
abstraction
\cite{DBLP:conf/concur/LaroussinieMS04, DBLP:conf/hybrid/BehrmannFHLPRV01}
to translate the problem to weighted (untimed) automata.

For looping automata, the above problems have been solved in \cite{DBLP:conf/formats/BouyerFLMS08}.
(This paper also treats energy games and so-called universal energy problems,
both of which are of no concern to us here.)
While we can re-use some of the methods of \cite{DBLP:conf/formats/BouyerFLMS08} for
our Büchi-enriched case, our extension is by no means trivial.
First, in the setting of \cite{DBLP:conf/formats/BouyerFLMS08} it suffices to find
\emph{any} reachable and energy positive loop;
now, our algorithm must consider that such loops might not be accepting in themselves
but give access to new parts of the automaton which are.
Secondly, \cite{DBLP:conf/formats/BouyerFLMS08} mostly treats the energy problem with unlimited upper bound,
whereas we consider that energy has a (``weak'') upper bound beyond which it cannot increase.

In \cite{DBLP:conf/formats/BouyerFLMS08} it is claimed that the weak-upper-bound problem can be solved
by slight modifications to their solution of the unbounded problem; but this is not the case.
For example, the typical Bellman-Ford detection of positive loops
might not work when the energy levels attained in the previous step
are already equal to the upper bound.
Moreover, we argue that the setting considering a weak upper bound is of greater
practical interest, as it allows to faithfully model actual physical systems
with a bounded capacity to store energy, such as electric vehicles.

As a second contribution, we have implemented all of our algorithms in a tool
based on the open-source platforms TChecker\footnote{See
  \url{https://github.com/ticktac-project/tchecker}}
\cite{DBLP:journals/iandc/HerbreteauSW16} and Spot\footnote{See
  \url{https://spot.lrde.epita.fr/}}
\cite{duret.22.cav} to solve $\omega$-regular energy problems for
one-clock weighted timed automata.
We first employ TChecker to compute the zone graph and then use this to construct the
corner-point abstraction.  This in turn is a weighted (untimed)
Büchi or Parity automaton, in which we also may apply a variant of
Alur and Dill's Zeno-exclusion
technique~\cite{DBLP:journals/tcs/AlurD94}.  Finally, our main
algorithm to solve the $\omega$-regular energy problems on weighted
finite automata is implemented using a fork of Spot.  Our software is
available at \url{https://github.com/PhilippSchlehuberCaissier/wspot}.

In our approach to solve the latter problem, we do not
and cannot fully separate the quantitative constraint on the energy
and qualitative acceptance condition (contrary to, for example,
\cite{DBLP:journals/tcs/ChatterjeeD12} which reduces energy parity games
to energy games).
We first determine the strongly connected
components (SCCs) of the unweighted automaton.
Then we treat each of the SCCs one by one depending on the acceptance condition.
In the case of a generalized Büchi accepting condition, we degeneralize it
using the standard counting construction \cite{DBLP:conf/cav/GastinO01}.
In the case of the Parity condition, we rely on an approach inspired
by a classical algorithm to solve Parity games, devised by Zielonka and
published in \cite{zielonka.98.tcs}; the approach presented in 
\cite{DBLP:journals/tcs/ChatterjeeD12} uses similar ideas (adapted to their setting).
The idea is to work in layers considering the highest priorities first:
If the highest priority is accepting, then we treat it much like an
accepting transition in the Büchi case.
If it is rejecting, we remove the corresponding transitions from the SCC and
search in the remaining graph.
Finally, we apply a modified Bellman-Ford algorithm to search for energy
feasible lassos that start on the main graph and loop on an accepting cycle
in the SCC.

This work is based on our contribution \cite{DBLP:conf/fm/DziadekFS23} and
extends it in two directions.
First we generalize the acceptance condition from Büchi to Parity, which allows
to represent the class of $\omega$-regular languages with deterministic automata.
Secondly we propose an efficient algorithm to compute the actual trace verifying
the quantitative and qualitative constraints.
This turns out to be a non-trivial task as the structure of these traces is
significantly more complicated than those for normal $\omega$-words.

We also correct an error in \cite{DBLP:conf/fm/DziadekFS23}.
There, it is claimed that two iterations suffice
to find Büchi accepting and energy feasible cycles.
We expose an example where more than two iterations are necessary
and where
the algorithm given in \cite{DBLP:conf/fm/DziadekFS23}
would fail to find energy feasible cycles.
In fact, the number of iterations necessary is linear in the weak upper bound
when using that approach.
We therefore devise a new algorithm which both corrects this mistake
and whose complexity does not depend on the weak upper bound.

The rest of the paper is structured as follows.
In Section 2 we introduce the energy Büchi problem for finite weighted automata.
We also show that these may be degeneralized and that searching for lassos is enough.
Section 3 introduces the energy Büchi problem for weighted timed automata
and shows how to reduce this to the problem for weighted (untimed) automata.
In Section 4 we finish solving the energy Büchi problem for finite weighted automata
by developing an algorithm to find feasible lassos.
Section 5 shows some benchmarks,
and in Section 6 we develop an algorithm to compute the actual trace verifying
the energy Büchi constraints in case it exists.
Section 7 shows how to reduce energy Parity problems to energy Büchi problems,
and Section 8 concludes.

\begin{rexample}
To clarify notation and put the concepts into context, we introduce a
small running example.  A satellite in low-earth orbit has a rotation
time of about 90 minutes, $40\%$ of which are spent in earth shadow.
Measuring time in minutes and
(electrical) energy in unspecified ``energy units'', we may thus model
its simplified base electrical system as shown in Figure~\ref{fig:ex.1}.

\begin{figure}[tbp]
\begin{subfigure}[t]{0.55\textwidth}
  \begin{center}
    \begin{tikzpicture}[x=.75cm]
    \node[state with output, initial left, rectangle split,
    rectangle split parts=2, rounded corners] (0) at (0,0) {$x\le
      35$ \nodepart{second} $-10$};
    \node[state with output, rectangle split,
    rectangle split parts=2, rounded corners] (1) at (5,0) {$x\le 55$
      \nodepart{second} $+40$};
    \path (0) edge[out=10, in=170] node {$x=35 \quad x\gets 0$}
    (1);
    \path (1) edge[out=-170, in=-10] node {$x=55 \quad x\gets 0$}
    (0);
  \end{tikzpicture}
  \subcaption{Weighted timed automaton $A_T$}
  \label{fig:ex.1}
  \end{center}
\end{subfigure}%
\begin{subfigure}[t]{0.44\textwidth}
  \begin{center}
    \begin{tikzpicture}[%
      state/.style={shape=rectangle, rounded corners, draw},
      y=1.4cm, x=1.4cm]
    \node[state, initial left] (1) {1};
    \node[state, right=of 1] (2) {2};
    \path (1) edge[bend left] node (t) {$-350$} (2);
    \path (2) edge[bend left] node {$2200$} (1);
  \end{tikzpicture}
  \subcaption{Equivalent (finite) weighted automaton $A$}
  \label{fig:ex.1.finite}
  \end{center}
\end{subfigure}
\caption{Satellite example: two representations of the base circuit.
We mark the acceptance condition of the automaton above its depiction.
Here both automata are in fact looping automata, as all infinite runs are accepted.
}
\label{fig:ex.1.both}
\end{figure}
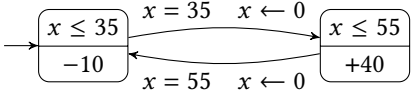
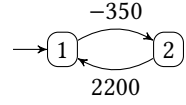

This is a weighted timed automaton (the formalism will be introduced
in Section~\ref{se:wta}) with one clock, $x$, and two locations.  The
clock is used to model time, which progresses with a constant rate but
can be reset on transitions.  The initial location on the left (modeling
earth shadow) is only active as long as $x\le 35$, and given that $x$
is initially zero, this means that the model may stay here
for at most $35$ minutes.  Staying in this location consumes $10$ energy
units per minute, corresponding to the satellite's base consumption.

After $35$ minutes the model transitions to the ``sun'' location on the
right, where it can stay for at most $55$ minutes and the solar panels
produce $50$ energy units per minute, from which the base consumption
has to be subtracted.  Note that the transitions can only be taken if
the clock shows exactly $35$ (resp. $55$) minutes; the clock is reset to zero
after the transition, as denoted by $x\gets 0$.  This ensures that the
satellite stays exactly $35$ minutes in the shadow and $55$ minutes in
the sun, roughly consistent with the physical reality.

Figure~\ref{fig:ex.1.finite} shows a translation of the automaton of Figure~\ref{fig:ex.1}
to a weighted untimed automaton.
State $1$ corresponds to the ``shadow'' location and transitions are
annotated with the corresponding weights, the rate of the location
multiplied by the time spent in it.  In Section~\ref{sec:cpa} we will
show how to obtain a weighted automaton from a weighted timed
automaton with precisely one clock.

One may now pose the following question: for a given battery capacity
$b$ and an initial charge $c$, is it possible for the satellite to
function indefinitely without ever running out of energy?  It is clear
that for $c<350$ or $b<350$, the answer is no: the satellite will run out of
battery before ever leaving Earth's shadow; for $b\ge 350$ and $c\ge 350$, it will indeed never
run out of energy.

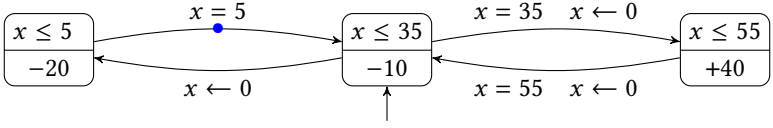
\begin{figure}[tbp]
    \centering
    \begin{tikzpicture}[x=.9cm]
      \path[use as bounding box] (-5,-.8) -- (5,.6);
      \node[state with output, initial below, rectangle split,
    rectangle split parts=2, rounded corners] (0) at (0,0) {$x\le
        35$ \nodepart{second} $-10$};
      \node[state with output, rectangle split, rectangle split parts=2,
        rounded corners] (1) at (5,0) {$x\le 55$
        \nodepart{second} $+40$};
      \path (0) edge[out=10, in=170] node {$x=35 \quad x\gets 0$}
      (1);
      \path (1) edge[out=-170, in=-10] node {$x=55 \quad x\gets 0$}
      (0);
      \node[state with output, rectangle split, rectangle split parts=2,
        rounded corners] (0a) at (-5,0) {$x\le 5\phantom{5}$
        \nodepart{second} $-20$};
      \path (0) edge[out=-170, in=-10] node {$x\gets 0$} (0a);
      \path (0a) edge[out=10, in=170] node {$x=5$}
      node[anchor=center] {$\color{blue}{\bullet}$} (0);
    \end{tikzpicture}
    \caption{Weighted timed Büchi automaton $A_{T1}$ for satellite with work module.
      Only infinite runs containing infinitely many transitions marked by
      $\color{blue}{\bullet}$ are accepted.
    }
    \label{fig:ex.2}
\end{figure}

Now assume that the satellite also has some work to do: once in a
while it must, for example, send some collected data to earth.  Given
that we can only handle weighted timed automata with precisely one clock
(see Section~\ref{se:wta}), we model the combined system as in
Figure~\ref{fig:ex.2}.
That is, work (modeled by the leftmost location) takes 5 minutes
and costs an extra $10$ energy units per minute.  The colored dot on
the outgoing transition of the work state marks a (transition-based)
Büchi condition which forces to take the transition infinitely
often in order for the run to be accepted.
As a consequence, all accepting runs also visit the ``work'' state indefinitely
often, consistent with the demand to send data once in a while.  In
order to model the system within the constraints of our modeling
formalism, we must make two simplifying assumptions, both unrealistic
but conservative:
\begin{itemize}
\item work occurs during earth shadow;
\item work prolongs earth shadow time.
\end{itemize}
The reason for the second property is that the clock $x$ is reset to
$0$ when entering the work state; otherwise we would not be able to
model that it lasts $5$ minutes without introducing a second clock.
It is clear how further work modules may be added in a similar way,
each with their own accepting color.

We will come back to this example later and, in particular, argue that
the above assumptions are indeed conservative
in the sense that any behavior admitted in our model is also present in a more realistic model which we will introduce.
\end{rexample}

\section{Energy Büchi Problems in Finite Weighted Automata}
\label{sec:wba}

We now define energy Büchi problems in finite weighted automata and
show how they may be solved.  The similar setting for weighted
\emph{timed} automata will be introduced in Section~\ref{se:wta}.

\begin{definition}[WBA]
  A \emph{weighted} (transition-based, generalized) \emph{Büchi
    automaton} (WBA) is a structure $A=(\mathcal{M}, S, s_0, T)$
  consisting of
 a finite set of colors $\mathcal{M}$,
 a set of states $S$ with initial state $s_0\in S$, and
 a set of transitions $T\subseteq S\times 2^\mathcal{M}\times \Real\times S$.
\end{definition}

A transition $t = (s, M, w, s')\in T$ in a WBA is thus
annotated by a set of colors $M$ and a real weight
$w$, denoted by $s\tto{w}_M s'$; to save ink, we may omit any or
all of $w$ and $M$ from transitions and $\mathcal{M}$ from WBAs.
The automaton $A$ is \emph{finite} if $S$ and
$T\subseteq S\times 2^\mathcal{M}\times \Int\times S$ are finite
(thus finite implies integer-weighted).
We use binary encoding for integer weights.

A \emph{run} in a WBA is a finite or infinite sequence
$\rho = s_1\to s_2\to\dotsm$.  We write $\first{\rho}=s_1$ for its
starting state and, if $\rho$ is finite, $\last{\rho}$ for its final
state.  \emph{Concatenation} $\rho_1 \rho_2$ of runs is the usual
partial operation defined if $\rho_1$ is finite and
$\last{\rho_1}=\first{\rho_2}$.  Also \emph{iteration} $\rho^n$ of
finite runs is defined as usual, for $\first{\rho}=\last{\rho}$, and
$\rho^\omega$ denotes infinite iteration.

For $c, b\in \Nat$ \footnote{Natural numbers include $0$.}  and a run
$\rho = s_1\tto{w_1} s_2\tto{w_2}\dotsm$, the
\emph{$(c, b)$-accumulated weights} of $\rho$ are the elements of the
finite or infinite sequence $\weights{c}{b}{\rho}=(e_1, e_2,\dotsc)$
defined by $e_1=\min(b, c)$ and $e_{i+1}=\min(b, e_i+w_{i})$.  Hence
the transition weights are accumulated, starting with $c$, but only up to the maximum
bound $b$; increases above $b$ are discarded.  We call $c$ the
\emph{initial credit} and $b$ the \emph{weak upper bound}.

\begin{rexample}
  In Figure~\ref{fig:ex.1.finite}, and choosing $c=360$ and $b=750$, we
  have a single infinite run
  $\rho =
  1\tto{-350}2\tto{2200}1\tto{-350}2\tto{2200}1\tto{-350}\dotsm$, with
  $\weights{c}{b}{\rho} = (360, 10, 750, 400, 750, \dotsc)$.
\end{rexample}

A run $\rho$ as above is said to be \emph{$(c, b)$-feasible} if
$\weights{c}{b}{\rho}_i\ge 0$ for all indices $i$, that is, the
accumulated weights of all prefixes are non-negative.  (This is the
case for the example run above.)
For a finite run $\rho=s_1\tto{w_1}\dotsm\to s_n$ we also write $\lweight{c\downarrow b}{\rho}=\weights{c}{b}{\rho}_n$ for its final accumulated weight,
and we will omit the ``${\downarrow} b$'' parts of this notation
if no confusion can arise.
The following simple fact will prove very useful later.

\begin{lemma}
  \label{le:lweight}
  For any finite run $\rho$ and $c_1, c_2, b\in \Nat$
  with $c_1\le c_2$,
  $\lweight{c_1}{\rho}\le \lweight{c_2}{\rho}\le \lweight{c_1}{\rho}+c_2-c_1$.
\end{lemma}

\begin{proof}
  For simplicity we may assume $c_2=c_1+1$.
  If there is an index $i$ such that $\sweights{c_1}{\rho}_i=b$,
  then $\lweight{c_2}{\rho}=\lweight{c_1}{\rho}$;
  otherwise, $\lweight{c_2}{\rho}=\lweight{c_1}{\rho}+1$.
\end{proof}

An infinite run $\rho = s_1\to_{M_1} s_2\to_{M_2}\dotsm$ is
(generalized, transition-based) \emph{Büchi accepted} if all colors in $\mathcal{M}$ are seen infinitely
often along $\rho$, that is,
$\mathcal{M}=\text{Inf}((M_i)_{i\geq 1})$ where
\[\text{Inf}((M_i)_{i\geq 1})=\{m\in \mathcal{M}\mid \forall i \in \Nat\ldotp
\exists j \in \Nat\ldotp j>i \text{ and }m\in M_j\}.\]

\begin{definition}
  The \emph{energy Büchi problem} for a finite WBA $A$, initial credit $c\in \Nat$
  and weak upper bound $b\in \Nat$ is to ask whether there
  exists a Büchi accepted $(c,b)$-feasible run in $A$.
\end{definition}

\begin{definition}
  The \emph{energy Büchi trace problem} for a finite WBA
  $A$ initial credit $c\in \Nat$ and weak upper bound $b\in \Nat$ is to
  extract a witness run respecting the quantitative and qualitative
  constraints given that the corresponding energy Büchi problem was
  answered positively.
\end{definition}

These definitions naturally extend to other acceptance conditions like parity.
A weighted (transition-based) Parity automaton (WPA) has the same structure
as a WBA, however the meaning of the colors changes.
Here, each color is assigned to a non-negative integer and we accept
all runs for which the largest color seen infinitely often is even.
This is commonly called a max-even-Parity condition, more details
on this are given in Sec.~\ref{sec:parity_cond}.

Energy problems for finite weighted automata without acceptance condition,
asking for the existence of \emph{any} infinite $c$-feasible run, have been introduced in
\cite{DBLP:conf/formats/BouyerFLMS08} and extended to multiple weight
dimensions in \cite{DBLP:conf/ictac/FahrenbergJLS11} where they are
related to vector addition systems and Petri nets.  We extend them to
(transition-based) generalized Büchi or Parity conditions here but do not
consider an extension to multiple weight dimensions.

\subsection*{Degeneralization}\label{sec:degen}

As a first step to solving energy problems for finite WBAs, we show that
the standard counting construction which transforms generalized Büchi
automata into simple Büchi automata with only one color, see for example
\cite{DBLP:conf/cav/GastinO01}, also applies in our weighted setting.
To see that, let $A=(\mathcal{M}, S, s_0, T)$ be a (generalized) WBA,
write $\mathcal{M}=\{m_1,\dotsc, m_k\}$, and define another WBA
$\bar{A}=(\bar{\mathcal{M}}, \bar{S}, \bar{s}_0, \bar{T})$ as follows:
\begin{gather*}
  \bar{\mathcal{M}} = \{m_a\} \qquad \bar{S} = S\times\{1,\dotsc, k\} \qquad \bar{s}_0 = (s_0, 1)
  \\
  \bar{T} =
  \begin{aligned}[t]
    & \big\{ ((s, i), \emptyset, w, (s', i)) \bigmid (s, M, w, s')\in T, m_i\notin M \big\} \\
    & \cup \big\{ ((s, i), \emptyset, w, (s', i+1)) \bigmid i\ne k, (s, M, w, s')\in T, m_i\in M \big\} \\
    & \cup \big\{ ((s, k), \{m_a\}, w, (s', 1)) \bigmid (s, M, w, s')\in T, m_k\in M \big\}
  \end{aligned}
\end{gather*}
That is, we split the states of $A$ into levels $\{1,\dotsc, k\}$.
At level $i$, the same transitions exist as in $A$, except those
colored with $m_i$; seeing such a transition puts us into level $i+1$,
or $1$ if $i=k$.  In the latter case, the transition in $\bar{A}$ is
colored by its only color $m_a$.
For succinctness we call such transitions back-edges, since they loop
back to the first level of the degeneralization and have a key role in our algorithm.
Intuitively, this preserves the language as we are sure that all
colors of the original automaton $A$ have been seen:

\begin{lemma}
  \label{le:degen}
  For any $c,b\in \Nat$, $A$ admits a Büchi accepted $(c,b)$-feasible run
  iff $\bar{A}$ does.
\end{lemma}

\begin{proof}
  Any infinite run $\rho$ in $A$ translates to an infinite run
  $\bar{\rho}$ in $\bar{A}$ by iteratively replacing transitions
  $(s, M, w, s')$ in $\rho$ with the corresponding transitions
  in $\bar{T}$.  Conversely, if $\bar{\rho}$ is an infinite run in
  $\bar{A}$, then we may replace any transition
  $((s, i), M, w, (s', j))$ in $\bar{\rho}$ with its
  preimage in $T$, yielding a run $\rho$
  in $A$.

  Given that the above construction does not affect the weights of
  transitions, it is clear that $\rho$ is $(c,b)$-feasible iff
  $\bar{\rho}$ is.  If $\rho$ is Büchi accepted, then $m_k$ is seen
  infinitely often along $\rho$, hence $m_a$ is seen infinitely often
  along $\bar{\rho}$ and $\bar{\rho}$ is Büchi accepted.  For the
  converse, assume that $\bar{\rho}$ is Büchi accepted, then $m_a$ is
  seen infinitely often along $\bar{\rho}$ and hence $m_k$ is seen
  infinitely often along $\rho$.

  To finish the proof, we have seen that $\bar{\rho}$ contains infinitely
  many transitions of the form $((s, k), \{m_a\}, x, (s', 1))$.  Hence
  $\bar{\rho}$ also contains infinitely many sequences of transitions
  \begin{equation*}
    (s_1, 1) \to\dotsm\to (s_1', 1) \to (s_2, 2) \to\dotsm\to (s_2',
    2)
    \to\dotsm\to (s_{k-1}', k-1) \to (s_k, k),
  \end{equation*}
  traversing all levels of $\bar{A}$.
  Each of these sequences corresponds, by construction, to a sequence
  of transitions in $\rho$ along which all of $m_1,\dotsc, m_{k-1}$
  are seen.  Hence all of $m_1,\dotsc, m_{k-1}$ have to be seen
  infinitely often along $\rho$ and it is thus Büchi accepted.
\end{proof}

\subsection*{Reduction to lassos}
\label{sec:RedLasso}

An infinite run $\rho$ in $A$ is~a \emph{lasso} if
$\rho=\gamma_1 \gamma_2^\omega$ for finite runs $\gamma_1$ and $\gamma_2$.
The following lemma shows that it suffices to search for lassos in
order to solve energy Büchi problems.

\begin{lemma}
  \label{le:lasso}
  For any $c,b\in \Nat$, $A$ admits a Büchi accepted $(c,b)$-feasible
  infinite run iff it admits a Büchi accepted $(c,b)$-feasible lasso.
\end{lemma}

\begin{proof}
  By degeneralization we may assume that $A$ has only one color.
  Let $\rho$ be a Büchi accepted $(c,b)$-feasible run in $A$.
  Assume first that there exist runs $\gamma_1$, $\gamma_2$ and $\rho'$
  (the first two finite and the last infinite)
  such that $\rho=\gamma_1 \gamma_2 \rho'$,
  $\first{\gamma_2}=\last{\gamma_2}$ (\ie $\gamma_2$ is a
  cycle), $\gamma_2$ visits an accepting transition,
  and $\lweight{0}{\gamma_2}\ge 0$.
  Using the first inequality of Lemma \ref{le:lweight},
  $\lweight{\bar c}{\gamma_2}\ge 0$ for any $\bar c\in \Nat$,
  so $\bar{\rho}=\gamma_1 \gamma_2^\omega$ is a Büchi accepted feasible lasso.

  Now assume that there is no cycle $\gamma_2$ as above.
  Let $t\in T$ be an accepting transition such that $\rho$ visits $t$ infinitely often
  and write $\rho=\gamma' t \gamma_1 t \gamma_2\dotsc$ as an infinite concatenation of finite runs.
  Then all $t \gamma_i$ are cycles which visit an accepting transition.
  By our assumption they must thus all satisfy
  $\lweight{0}{t \gamma_i}<0$, \ie $\lweight{0}{t \gamma_i}\le -1$.
  Using the second inequality of Lemma \ref{le:lweight},
  $\lweight{\bar c}{t \gamma_i}\le \bar c-1$ for all $\bar c\in \Nat$.
  Let $c'=\lweight{c}{\gamma'}$, then
  $\lweight{c}{\gamma' t \gamma_1\dotsc t \gamma_{c'+1}}<0$ in
  contradiction to $c$-fea\-si\-bil\-i\-ty of~$\rho$.
\end{proof}

Hence our energy Büchi problem may be solved by searching for Büchi
accepted $c$-feasible lassos.  We detail how to do this in
Section~\ref{sec:implementation}, here we just sum up the complexity
result which we prove at the end of Section~\ref{se:impl}.

\begin{theorem}
  \label{th:wba}
  Energy Büchi problems for finite WBA are decidable in polynomial
  time.
\end{theorem}

\section{Energy Büchi Problems for Weighted Timed Automata}
\label{se:wta}

We now extend our setting to weighted timed automata.
Let $X$ be a finite set of clocks.  We denote by $\Phi( X)$ the set of
\emph{clock constraints} $\phi$ on $X$, defined by the following grammar:
\begin{equation*}
  \phi ::= x\bowtie k\mid \phi_1\wedge \phi_2
  \qquad\quad
  \big( x\in X, k\in \Nat, {\bowtie}\in\{{\le}, {<}, {\ge}, {>}, {=}\} \big)
\end{equation*}
A \emph{clock valuation} on $X$ is a function $v: X\to \Realnn$.  The
clock valuation $v_0$ is given by $v_0(x)=0$ for all $x\in X$, and
for $v: X\to \Realnn$, $d\in \Realnn$, and
$R: X\to (\Nat\cup\{\bot\})$, we define the delay $v+d$ and reset
$v[R]$ by
\begin{equation*}
  (v+d)(x)= v(x)+d, \qquad%
  v[R](x)=
  \begin{cases}
    v(x) &\text{if } R(x)=\bot, \\
    R(x) &\text{otherwise}.
  \end{cases}
\end{equation*}
Note that in $v[R]$ we allow clocks to be reset to arbitrary
non-negative integers instead of only $0$ which is assumed in most of
the literature.  It is known \cite{DBLP:journals/sttt/LarsenPY97} that
this does not change expressivity, but it adds notational convenience.
A clock valuation $v$ \emph{satisfies} clock constraint $\phi$,
denoted $v\models \phi$, if $\phi$ evaluates to true with $x$ replaced
by $v(x)$ for all $x\in X$.

\begin{definition}[WTBA]
  A \emph{weighted timed} (transition-based, generalized) \emph{Büchi
    automaton} (WTBA) is a structure
  $A=(\mathcal{M}, Q, q_0, X, I, E, r)$ consisting of
 a finite set of colors $\mathcal{M}$,
 a finite set of locations $Q$ with initial location $q_0\in Q$,
 a finite set of clocks $X$,
 location invariants $I: Q\to \Phi(X)$,
 a finite set of edges $E\subseteq Q\times 2^\mathcal{M}\times \Phi(X)\times (\Nat\cup\{\bot\})^X\times Q$, and
 location weight-rates $r: Q\to \Int$.
\end{definition}

As before, we may omit $\mathcal{M}$ from the signature and colors from edges if they are not necessary in the context.
Note that the edges carry no weights here, which would correspond to discrete
weight updates. In a WTBA, only locations are weighted
by a rate.  Even without Büchi conditions, the approach laid out here
would not work for weighted edges.  This was already noted in
\cite{DBLP:conf/formats/BouyerFLMS08}; instead it requires different
methods which are developed in \cite{DBLP:conf/hybrid/BouyerFLM10}
(see also \cite{DBLP:journals/actaC/EsikFLQ17,
  DBLP:journals/actaC/EsikFLQ17a}).  There, one-clock weighted timed
automata (with edge weights) are translated to finite automata
weighted with so-called \emph{energy functions} instead of integers.
We believe that our extension to Büchi conditions should also work in
this extended setting, but leave the details to future work.

The \emph{semantics} of a WTBA $A$ as above is the (infinite) WBA
$\sem{A}=(\mathcal{M}, S, s_0, T)$ given by
$S=\{(q, v)\in Q\times \Realnn^X\bigmid v\models I(q)\}$ and
$s_0=(q_0, v_0)$.
Transitions in $T$ are of the following two types:
\begin{itemize}
\item \emph{delays} $(q, v)\smash{\tto{w}}^d_\emptyset (q, v+d)$ for all
  $(q, v)\in S$ and $d\in \Realnn$ for which $v+d'\models I(q)$ for
  all $d'\in [0, d]$, with $w=r(q) d$;
  \footnote{Here we annotate transitions with the time $d$ which passes; we only need this to exclude Zeno runs below and will otherwise omit the annotation.}

  \smallskip
\item \emph{switches} $(q, v)\smash{\tto{0}}^0_{M} (q', v')$ for all
  $e=(q, M, g, R, q')\in E$ for which $v\models g$, $v'=v[R]$
  and $v'\models I(q')$.
\end{itemize}

Each state in $\sem{A}$ corresponds to a tuple containing
a location in $A$ and a clock valuation $X\to \Realnn$. This allows
to keep track of the discrete state as well as the evolution of the clocks.
By abuse of notation, we will sometimes write $(q, v)\in \sem{A}$
instead of $(q, v)\in S$, for $S$ as defined above.

We may now pose energy Büchi problems also for WTBAs,
but we wish to exclude infinite runs in which time is bounded, so-called
Zeno runs.
Formally an infinite run $(q_0, v_0)\to^{d_1} (q_1, v_1)\to^{d_2}\dotsm$ is
\emph{Zeno} if $\sum d_i$ is finite:
Zeno runs admit infinitely many steps in finite time and
are hence considered unrealistic from a modeling point of view
\cite{DBLP:journals/tcs/AlurD94, DBLP:journals/corr/abs-1106-1850}.

\begin{definition}
  The \emph{energy Büchi problem} for a WTBA $A$, initial credit $c\in \Nat$ and weak upper bound $b\in \Nat$ is to ask if there exists a Büchi
  accepted $(c,b)$-feasible non-Zeno run in $\sem{A}$.
\end{definition}

We continue our running example; but to do so properly, we need to introduce products of WTBAs.
Let $A_i=(\mathcal{M}_i, Q_i, q_0^i, X_i, I_i, E_i, r_i)$, for $i\in\{1, 2\}$, be WTBAs.
Their \emph{product} is the WTBA $A_1\mathbin{\|} A_2=(\mathcal{M}, Q, q_0, X, I, E, r)$ with
\begin{gather*}
  \mathcal M=\mathcal M_1\cup \mathcal M_2, \qquad Q=Q_1\times Q_2, \qquad q_0=(q_0^1, q_0^2), \qquad
  X=X_1\cup X_2, \\
  I((q_1, q_2))=I(q_1)\land I(q_2), \qquad r((q_1, q_2))=r(q_1)+r(q_2), \\
  \begin{aligned}
    E ={} &\big\{ ((q_1, q_2), M, g, R, (q_1', q_2)) \bigmid (q_1, M, g, R, q_1')\in E_1 \big\} \\
    {}\cup{} &\big\{ ((q_1, q_2), M, g, R, (q_1, q_2')) \bigmid (q_2, M, g, R, q_2')\in E_2 \big\}.
  \end{aligned}
\end{gather*}

\begin{rexample}
Let $A$ be the basic WTBA of Figure~\ref{fig:ex.1}
and $A_1$ the combination of $A$ with the work module of
Figure~\ref{fig:ex.2}.
Now, instead of building $A_1$ as we have done, a principled way of constructing a model for the satellite-with-work-module would be to first model the work module $W$ and then form the product $A\mathbin{\|} W$.
We show such a work module and the resulting product $B_1$ in Figure~\ref{fig:ex.work.combi}.

\begin{figure}[tbp]
\begin{subfigure}[t]{0.3\textwidth}
  \centering
  \begin{tikzpicture}[x=.9cm, y=.5cm,baseline=(0a.south)]
    \node[state with output, initial left, rectangle split,
    rectangle split parts=2, rounded corners] (0) at (0,0)
    { \nodepart{second} $0$};
    \node[state with output, rectangle split,
    rectangle split parts=2, rounded corners] (0a) at (0,-5) {$y\le 5$
      \nodepart{second} $-10$};
    \path (0) edge[bend right=10] node[left] {$y\gets 0$} (0a);
    \path (0a) edge[bend right=10] node[right] {$y=5$}
    node[anchor=center] {$\color{blue}\bullet$} (0);
  \end{tikzpicture}
  \subcaption{}
  \label{fig:ex.work}
\end{subfigure}
\begin{subfigure}[t]{0.68\textwidth}
  \centering
  \begin{tikzpicture}[x=1.1cm, y=.5cm,baseline=(2.south)]
    \node[state with output, initial left, rectangle split,
    rectangle split parts=2, rounded corners] (0) at (0,0) {$x\le
      35$ \nodepart{second} $-10$};
    \node[state with output, rectangle split, rectangle split
    parts=2, rounded corners] (1) at (5,0) {$x\le 55$
      \nodepart{second} $+40$};
    \path (0) edge[out=10, in=170] node {$x=35 \quad x\gets 0$}
    (1);
    \path (1) edge[out=-170, in=-10] node {$x=55 \quad x\gets 0$}
    (0);
    \node[state with output, rectangle split, rectangle split
    parts=2, rounded corners] (2) at (0,-5) {$x\le
      35\land y\le 5$ \nodepart{second} $-20$};
    \node[state with output, rectangle split, rectangle split
    parts=2, rounded corners] (3) at (5,-5) {$x\le 55\land y\le 5$
      \nodepart{second} $+30$};
    \path (2) edge[out=10, in=170] node {$x=35 \quad x\gets 0$}
    (3);
    \path (3) edge[out=-170, in=-10] node {$x=55 \quad x\gets 0$}
    (2);
    \path (0) edge[out=260, in=100] node[swap] {$y\gets 0$} (2);
    \path (2) edge[out=80, in=280] node[swap] {$y=5$}
    node[anchor=center] {$\color{blue}\bullet$} (0);
    \path (1) edge[out=260, in=100] node[swap] {$y\gets 0$} (3);
    \path (3) edge[out=80, in=280] node[swap] {$y=5$}
    node[anchor=center] {$\color{blue}\bullet$} (1);
  \end{tikzpicture}
  \subcaption{}
  \label{fig:ex.prod}
\end{subfigure}
\caption{Satellite example. (a) work module $W$;
  (b) product $B_1=A\mathbin{\|} W$}
\label{fig:ex.work.combi}
\end{figure}
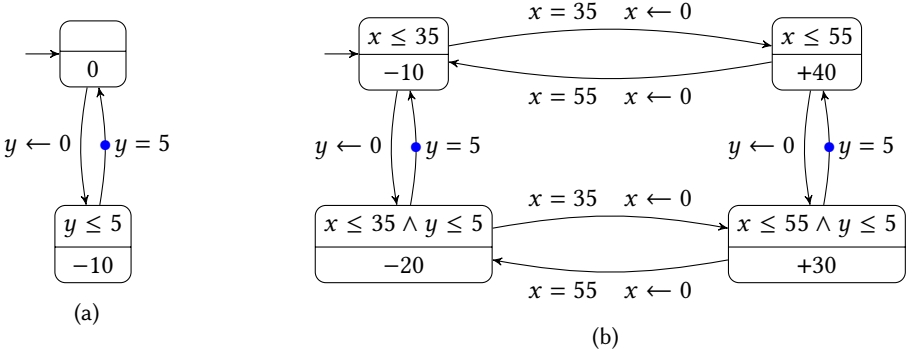

As expected, $W$ expresses that work takes $5$ minutes and costs $10$ energy units per minute,
and the Büchi condition enforces that work is executed infinitely often.
The product $B_1$ models the shadow-sun phases together with the fact that work may be executed at any time,
and contrary to our ``unrealistic'' model $A_1$ of Figure~\ref{fig:ex.2},
work does not prolong earth shadow time.

Now $B_1$ has \emph{two} clocks, and we will see below that our
constructions can handle only one.  This is the reason for our
``unrealistic'' model $A_1$,
and we can now state precisely in which sense it is conservative:
if $\sem{A_1}$ admits a Büchi accepted $c$-feasible non-Zeno run, then so does $\sem{B_1}$.
For a proof of this fact, one notes that any infinite run $\rho$ in $\sem{A_1}$ may be translated to an infinite run $\bar{\rho}$ in $\sem{B_1}$ by adjusting the clock valuation by $5$ whenever the work module is visited.
\end{rexample}

\subsection*{Bounding Clocks}

As a first step to solve energy Büchi problems for WTBAs,
we show that we may assume that the clocks in any WTBA~$A$ are bounded
above by some $N\in \Nat$, \ie such that $v(x)\le N$ for all
$(q, v)\in \sem{A}$ and $x\in X$.  This is shown for reachability in
\cite{DBLP:conf/hybrid/BehrmannFHLPRV01}; the following lemma extends it to Büchi acceptance.

\begin{lemma}
  \label{le:clockbound}
  Let $A=(\mathcal{M}, Q, q_0, X, I, E, r)$ be a WTBA and $c,b\in \Nat$.
  Let $N$ the maximum constant appearing in any invariant $I(q)$, for
  $q \in Q$, or in any guard $g$,
  for $(q, M, g, R, q')\in E$.  There is a WTBA
  $\bar{A}=(\mathcal{M}, Q, q_0, X, \bar{I}, \bar{E}, r)$ such that
  \begin{enumerate}
  \item $v(x)\le N+2$ for all $x\in X$ and $(q, v)\in \sem{\bar{A}}$, and
  \item there exists a $(c,b)$-feasible Büchi accepted run in
    $\sem{A}$ iff such a run exists in~$\sem{\bar{A}}$.
  \end{enumerate}
\end{lemma}

\begin{proof}
  Following \cite{DBLP:conf/hybrid/BehrmannFHLPRV01} we define
  \begin{gather*}
    \bar{E} = E\cup \big\{ (q, \emptyset, (x=N+2), (x\gets
    N+1), q) \bigmid q\in Q, x\in X \big\}, \\
    \textstyle \bar{I}(q) = I(q) \land \bigland_{x\in X} (x\le N+2),
  \end{gather*}
  that is, clock values are reset to $N+1$ whenever they reach $N+2$
  using uncolored transitions.  Hence $\bar{A}$ satisfies
  the first requirement.

  Let ${\cong}\subseteq \Realnn^X\times \Realnn^X$ be the
  relation on clock valuations defined by
  \begin{equation*}
    v\cong v' \quad\text{iff}\quad
    \forall x\in X:
    v(x)\le N\implies
    v(x)=v'(x), v(x)>N\iff v'(x)>N.
  \end{equation*}
  Using the same proof as in \cite{DBLP:conf/hybrid/BehrmannFHLPRV01},
  we may show that for any states $(q, v)$, $(q', v')$ in $\sem{A}$
  and any finite run $\rho: (q, v)\tto{w_1}\dotsm\tto{w_n} (q', v')$,
  there exists a finite run
  $\bar{\rho}: (q, v)\tto{w_1'}\dotsm\tto{w_m'} (q', v'')$ in
  $\sem{\bar{A}}$ with $v'\cong v''$ and
  $\lweight{c}{\rho} = \lweight{c}{\bar{\rho}}$, and vice versa:
  $\bar{\rho}$ is constructed from $\rho$ by inserting special reset
  transitions when appropriate, and $\rho$ from $\bar{\rho}$ by
  removing them.

  Using the above procedure, any infinite run $\rho$ in $\sem{A}$ may
  be iteratively converted to an infinite run $\bar{\rho}$ in
  $\sem{\bar{A}}$ and vice versa.  It is clear that $\rho$ is Büchi
  accepted iff $\bar{\rho}$ is, and that $\rho$ is $(c, b)$-fea\-si\-ble iff
  $\bar{\rho}$~is.
\end{proof}

\subsection*{Corner-point abstraction}
\label{sec:cpa}

We now restrict to WTBAs with only \emph{one} clock
and show how to translate these into finite untimed WBAs using the corner-point abstraction.
This abstraction
may be defined for any number of clocks, but it is shown in
\cite{DBLP:journals/pe/BouyerLM14} that the energy problem is
undecidable for weighted timed automata with four clocks or more; for two or three clocks
the problem is open.

Let $A=(\mathcal{M}, Q, q_0, X, I, E, r)$ be a WTBA with $X=\{x\}$ a
singleton.  Using Lemma \ref{le:clockbound} we may assume that $x$ is
bounded by some $N\in \Nat$, \ie such that $v(x)\le N$ for all
$(q, v)\in \sem{A}$.

Let $\mathfrak{C}$ be the set of all constants which occur in invariants $I(q)$
or guards $g$ or resets $R$ of edges $(q, M, g, R, q')$ in $A$,
and write $\mathfrak{C}\cup\{N\}=\{a_1,\dotsc, a_{n+1}\}$ with ordering
$0\le a_1<\dotsm< a_{n+1}$.
The \emph{corner-point regions}
\cite{DBLP:conf/concur/LaroussinieMS04,
  DBLP:conf/hybrid/BehrmannFHLPRV01} of $A$ are the
subsets
 $\{a_i\}$, for $i=1,\dotsc, n+1$, $[a_i, a_{i+1}\oir$, and $\oil a_i, a_{i+1}]$, for $i=1,\dotsc, n$,
of $\Realnn$;
that is, points, left-open, and right-open intervals on
$\{a_1,\dotsc, a_{n+1}\}$.

These are equivalent to clock constraints
$x=a_i$, $a_i\le x<a_{i+1}$, and $a_i<x\le a_{i+1}$, respectively,
defining a notion of implication $\cpr\implies \phi$ for $\cpr$ a
corner-point region and $\phi\in \Phi(\{x\})$.

The corner-point abstraction of $A$ is the finite WBA
$\cpa{A}=(\mathcal{M}\cup\{m_z\}, S, s_0, T)$, where
$m_z\notin \mathcal{M}$ is a new color,
$S=\{(q, \cpr)\mid q\in Q, \cpr\text{ corner-point}$ $\text{region of } A, \cpr\implies
  I(q)\}$,
$s_0=(q_0, \{0\})$, and transitions in $T$ are of the following types:
\begin{itemize}
\item \emph{delays}
  $(q, \{a_i\})\tto{0}_\emptyset (q, [a_i, a_{i+1}\oir)$,
  $(q, [a_i, a_{i+1}\oir)\tto{w}_{\{m_z\}} (q, \oil a_i,
  a_{i+1}])$ with $w=r(q)(a_{i+1}-a_i)$, and
  $(q, \oil a_i, a_{i+1}])\tto{0}_\emptyset (q, a_{i+1})$;
\item \emph{switches} $(q, \cpr)\tto{0}_{M} (q', \cpr)$ for
  $e=(q, M, g, (x\mapsto \bot), q')\in E$ with $\cpr\implies g$ and
  $(q, \cpr)\tto{0}_{M} (q', \{k\})$ for
  $e=(q, M, g, (x\mapsto k), q')\in E$ with $\cpr\implies g$.
\end{itemize}

The new color $m_z$ is used to rule out Zeno runs, see
\cite{DBLP:journals/tcs/AlurD94} for a similar construction: any Büchi
accepted infinite run in $\cpa{A}$ must have infinitely many
time-increasing delay transitions
$(q, [a_i, a_{i+1}\oir)\tto{w}_{\{m_z\}} (q, \oil a_i, a_{i+1}])$.

\begin{theorem}
  \label{th:cpa}
  Let $A$ be a one-clock WTBA and $c\in \Nat$.
  \begin{enumerate}
  \item If there is a non-Zeno Büchi accepted $c$-feasible run in
    $\sem{A}$, then there is a Büchi accepted $c$-feasible run in
    $\cpa{A}$.
  \item If there is a Büchi accepted $c$-feasible run in $\cpa{A}$,
    then there is a non-Zeno Büchi accepted $(c+\epsilon)$-feasible run
    in $\sem{A}$ for any $\epsilon>0$.
  \end{enumerate}
\end{theorem}

The so-called \emph{infimum energy condition}
\cite{DBLP:conf/formats/BouyerFLMS08} in the second part above,
replacing $c$ with $c+\epsilon$, is necessary in the presence of
\emph{strict} constraints $x<c$ or $x>c$ in $A$.  The proof maps runs
in $A$ to runs in $\cpa{A}$ by pushing delays to endpoints of
corner-point regions, ignoring strictness of constraints, and this has
to be repaired by introducing the infimum condition.

\begin{proof}
  To show the first part, let $\rho$ be a non-Zeno Büchi accepted
  $c$-feasible run in $\sem{A}$.  We follow
  \cite{DBLP:conf/formats/BouyerFLMS08} and convert $\rho$ to an
  infinite run $\bar{\rho}$ in $\cpa{A}$ by pushing delay transitions
  within a common corner-point region to the most profitable
  endpoints.  Let $k\in\{1,\dotsc, n\}$ and consider a maximal
  subsequence
  \begin{equation*}
    (q_1, v_1)\tto{w_1}\dotsm\tto{w_{m-1}}(q_m, v_m)
  \end{equation*}
  of $\rho$ for which $v_1,\dotsc, v_m\in \oil a_k, a_{k+1}\oir$.
  Then all switch transitions in this sequence are non-resetting.

  Let $(q'_1,\dotsc, q'_p)$ be the subsequence of $(q_1,\dotsc, q_m)$
  of first-unique elements, that is,
  $(q'_1,\dotsc, q'_p)=(q_{i_1},\dotsc, q_{i_p})$ is such that
  $q_{i_1} = q_1 =\dotsm= q_{i_2-1} \ne q_{i_2} =\dotsm= q_{i_3-1}$
  etc.  Let $j\in \{1,\dotsc, p\}$ be such that $r(q'_j)$ is maximal,
  then we construct the following finite run in $\cpa{A}$:
  \begin{equation*}
    (q'_1, [a_k, a_{k+1}\oir) \tto{0}\dotsm\tto{0} (q'_j, [a_k,
    a_{k+1}\oir)
    {}\tto{w} (q'_j, \oil a_k, a_{k+1}])
    \tto{0}\dotsm\tto{0} (q'_p, \oil a_k, a_{k+1}]),
  \end{equation*}
  with $w=r(q'_j)\, (a_{k+1}-a_k)$.  (This is possible as there are no
  guards or invariants in the interval $]a_k, a_{k+1}[$.)

  We have seen how to convert maximal finite non-resetting sub-runs of
  $\rho$ to runs in $\cpa{A}$, so all we have left to treat are
  resetting switch transitions $(q, v)\tto{0} (q', x\gets k)$; these
  are converted to transitions $(q, \cpr)\tto{0} (q', \{k\})$.  This
  finishes the construction of $\bar{\rho}$.

  Now $\bar{\rho}$ sees all the colors in $\mathcal{M}$ infinitely often because
  $\rho$ does.  Given that we have maximized energy gains when
  converting $\rho$ to $\bar{\rho}$ (by pushing delays to the most
  profitable location), $\bar{\rho}$ is also $c$-feasible.  Finally, $\rho$
  being non-Zeno implies that also the color $m_z$ is seen infinitely
  often in $\bar{\rho}$.

  For the other direction, let $K$ be the maximum absolute value of
  the weight-rates of $A$ and $\bar{\rho} = (q_0, \{0\}) t_1' t_2'\dotsm$ a
  Büchi accepted $c$-feasible run in $\cpa{A}$.  We
  iteratively construct an infinite run
  $\rho = (q_0, 0) t_1 t_2\dotsm$ in $\sem{A}$; we have to be careful
  with region boundaries because of potential strict constraints in $A$.

  Assume that $t_1\dotsm t_{n-1} (q, v)$ has been constructed.
  \begin{itemize}
  \item If $t_n'$ is a switch $(q, \cpr)\tto{0}_{M} (q', \cpr)$, we let
    $t_n=(q, v)\tto{0}_{M} (q', v)$.
  \item If $t_n'$ is a switch $(q, \cpr)\tto{0}_{M} (q', \{k\})$, we
    let $t_n=(q, v)\tto{0}_{M} (q', v[x\mapsto k])$.
  \item If $t_n'$ is a delay
    $(q, \{a_i\})\tto{0}_\emptyset (q, [a_i, a_{i+1}\oir)$, then
    $v(x)=a_i$ and we let $t_n=(q, v)\tto{w}_\emptyset (q, v+d)$
    with $d=\frac{\epsilon}{2^{n+1} K}$ and $w=r(q) d$.  (We must
    introduce a delay $\frac{\epsilon}{2^{n+1} K}$ here given that the
    constraint in $q$ may be strict; but we keep it sufficiently small
    so that in the end, the sum of all such new delays is bounded
    above.)
  \item If $t_n'$ is a delay
    $(q, [a_i, a_{i+1}\oir)\tto{w'}_{\{m_z\}} (q, \oil a_i,
    a_{i+1}])$, then by construction,
    $v(x) = a_i + \frac{\epsilon}{2^m K}$ for some $1\le m\le n$, and
    we let $t_n=(q, v)\tto{w}_\emptyset (q, v+d)$ with
    $d = a_{i+1} - a_{i} - \frac{\epsilon}{2^{n+1} K} -
    \frac{\epsilon}{2^m K}$ and $w=r(q) d$.
  \item If $t_n'$ is a delay
    $(q, \oil a_i, a_{i+1}])\tto{0}_\emptyset (q, a_{i+1})$, then
    by construction, $v(x) = a_{i+1} - \frac{\epsilon}{2^m K}$ for some
    $1\le m\le n$, and we let $t_n=(q, v)\tto{w}_\emptyset (q,
    v+d)$ with $d = \frac{\epsilon}{2^m K}$ and $w=r(q) d$.
  \end{itemize}

  If is clear that $\rho$ is Büchi accepted.  Given that $\bar{\rho}$ is
  $c+\epsilon$-feasible and the differences in delays between $\bar{\rho}$
  and $\rho$ are bounded by
  $\sum_{n=1}^\infty \frac{\epsilon}{2^n K} = \frac{\epsilon}{K}$, it
  is clear that $\rho$ is $c+\epsilon$-feasible.
\end{proof}

\begin{figure}[tbp]
  \centering
  \begin{tikzpicture}[state/.style={shape=rectangle, rounded
      corners, draw}, y=1.4cm, x=1.8cm]
    \node[state, initial left] (00) at (1.3,1) {$\{0\}$};
    \node[state] (01) at (2.3,1) {$[0, 35\oir$};
    \node[state] (02) at (3.7,1) {$\oil 0, 35]$};
    \node[state] (03) at (4.7,1) {$\{35\}$};
    \path (00) edge (01);
    \path (01) edge node[swap] {$-350$} node[anchor=center]
    {$\color{orange} \bullet$} (02);
    \path (02) edge (03);
    \node[state] (10) at (0,0) {$\{0\}$};
    \node[state] (11) at (1,0) {$[0, 35\oir$};
    \node[state] (12) at (2.1,0) {$\oil 0, 35]$};
    \node[state] (13) at (3,0) {$\{35\}$};
    \node[state] (14) at (3.9,0) {$[35, 55\oir$};
    \node[state] (15) at (5,0) {$\oil 35, 55]$};
    \node[state] (16) at (6,0) {$\{55\}$};
    \path (10) edge (11);
    \path (11) edge node {$1400$} node[anchor=center]
    {$\color{orange} \bullet$} (12);
    \path (12) edge (13);
    \path (13) edge (14);
    \path (14) edge node {$800$} node[anchor=center]
    {$\color{orange} \bullet$} (15);
    \path (15) edge (16);
    \path (03.south west) edge (10.north east);
    \path (16.north west) edge (00.south east);
  \end{tikzpicture}
  \caption{Corner-point abstraction of base module of Figure~\ref{fig:ex.1}.}
  \label{fig:ex.cpa}
\end{figure}
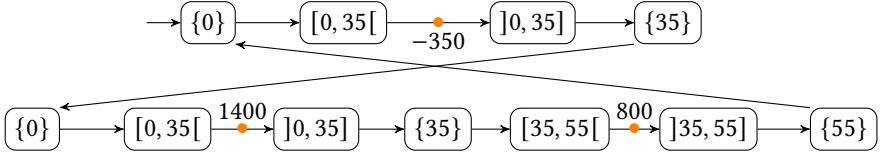

\begin{rexample}
We construct the corner-point abstraction of the base module $A$ of
Figure~\ref{fig:ex.1}.  Its constants are $\{0, 35, 55\}$, yielding the
following corner point regions:
\begin{equation*}
  \{0\}, \quad [0, 35\oir, \quad \oil 0, 35], \quad \{35\}, \quad
  [35, 55\oir, \quad \oil 35, 55], \quad \{55\}
\end{equation*}
The corner-point abstraction of $A$ now looks as in
Figure~\ref{fig:ex.cpa}, with the states corresponding to the ``shadow''
location in the top row; the colored transitions correspond to
the ones in which time elapses.
Note that this WBA is equivalent to the one in Figure~\ref{fig:ex.1.finite}.
\end{rexample}

Using the corner-point abstraction, we may now solve energy Büchi problems for one-clock WTBAs
by translating them into finite WBAs
and applying the algorithms of Section~\ref{sec:wba} and the forthcoming Section~\ref{se:impl}.
Note that as we only have one clock, the size of the corner-point abstraction is linear in the size of the input WTBA.

\section{Implementation}
\label{se:impl}
\label{sec:implementation}

We now describe our algorithm to solve energy Büchi problems for
finite WBA before detailing the changes necessary to treat Parity automata.
All of this has been implemented and is available at
\url{https://github.com/PhilippSchlehuberCaissier/wspot}.

We have seen in Section~\ref{sec:RedLasso} that this problem is
equivalent to the search for Büchi accepted $(c,b)$-feasible lassos.
By definition, a lasso $\rho=\gamma_1 \gamma_2^\omega$ consists of two parts, the
lasso prefix $\gamma_1$ (possibly empty, only traversed once) and the lasso
cycle $\gamma_2$ (repeated indefinitely).
In order for $\rho$ to be Büchi accepted and $(c,b)$-feasible the following
constraints need to hold:
\begin{itemize}
  \item the prefix $\gamma_1$ must be $(c,b)$-feasible;
  \item the cycle $\gamma_2$ must be $(\lweight{c}{\gamma_1},b)$-feasible;
  \item the cycle $\gamma_2$ must be $(\lweight{c}{\gamma_1\gamma_2^i},b)$-feasible for all $i > 0$.
\end{itemize}

The first constraint ensures that the prefix is energy feasible.
The second constraint ensures that we can take the cycle once after
traversing the prefix.
The third constraint expresses the need to loop in $\gamma_2$ indefinitely.
In fact, the energy $\lweight{c}{\gamma_1}$ may be greater than
$\lweight{c}{\gamma_1\gamma_2}$, however after sufficiently many
traversals of $\gamma_2$ the energy must stabilize (that is
$\lweight{c}{\gamma_1\gamma_2^m} = \lweight{c}{\gamma_1\gamma_2^{m+1}}$
for some sufficiently large $m$), while remaining $(c,b)$-feasible at all times.
As it turns out it is quite tricky to get this part correct and we will discuss
it in greater detail later on.

Finally the cycle $\gamma_2$ obviously needs to be Büchi accepted.

\begin{algorithm}[bt]
\caption{Algorithm to find Büchi accepted lassos in WBA\label{alg:second_part}}
\begin{algorithmic}[1]
\item[]\noindent \hskip-\leftmargin\textbf{Input:} weak upper bound $b$
\Function{BüchiEnergy}{graph $G$, initial credit $c$}
  \State $E \gets$ \Call{FindMaxE}{$G, G.\mathit{initial\_state}, c$}
    \Comment{$E\colon S\rightarrow\Nat$, mapping states to energy}
  \State $\mathit{SCCs} \gets$ \Call{Couvreur}{$G$} \Comment{Find all SCCs}
  \ForAll{$\mathit{scc} \in \mathit{SCCs}$}
    \State $\mathit{GS}, \mathit{back}\text{-}\mathit{edges} \gets degeneralize(\mathit{scc})$
    \ForAll{$be=\mathit{src}\xrightarrow{w} \mathit{dst} \in \mathit{back}\text{-}\mathit{edges}$}
        \State $E' \gets$ \Call{FindMaxE}{$\mathit{GS}, \mathit{dst},
          E[dst]$} \Comment{$be.dst$ is in $G$ and $GS$...}
        \label{algBE:firstIter}
        \State $e' \gets \min(b, E'[\mathit{src}]+w)$ \Comment{...(see Figure~\ref{fig:double_checking_degen})}
        \If{$E[\mathit{dst}] \le e'$} \Return True
        \Else
          \Comment{Second iteration (see Ex.~\ref{ex:double_checking})}
          \State $E'' \gets$ \Call{FindMaxE}{$\mathit{GS}, \mathit{dst}, e'$}
          \State $e'' \gets \min(b, E''[\mathit{src}]+w)$
          \If{$e' \le e''$} \Return True
          \Else
            \Comment{More iterations (see Ex.~\ref{ex:triple_check})}
            \ForAll{states $s_M$ with $E''[s_M]=b$}
              \label{alg:second_part:more_iterations}
              \State $E_\to \gets$ \Call{FindMaxE}{$\mathit{GS}, s_M, b$}
              \State $e_\text{dst} \gets \min(b, E_\to[\mathit{src}]+w)$
              \State $E_\leftarrow \gets$ \Call{FindMaxE}{$\mathit{GS}, dst, e_\text{dst}$}
              \If{$E_\leftarrow[s_M] = b$} \Return True
              \EndIf
            \EndFor
          \EndIf
        \EndIf
      \EndFor
    \EndFor
    \State \Return False
\EndFunction
\end{algorithmic}
\end{algorithm}

\subsection*{Finding lassos}

The overall procedure to find lassos is described in
Algorithm~\ref{alg:second_part}. It is based on two steps.  In step
one we compute all energy-optimal paths starting at the initial state
of the automaton with initial credit $c$.  This step is done on the
original WBA, and we do not take into account the colors. Optimal
paths found in this step will serve as lasso prefixes.

The second step is done individually for each Büchi accepting SCC.
The \textsc{Couvreur} algorithm, used to identify the SCCs, ignores the weights,
and we can use the version distributed by Spot.
We then degeneralize the accepting SCCs one by one,
as described in Section~\ref{sec:degen}; recall that
this creates one copy of the SCC, which we call a level, per
color. The first level roots the degeneralization in the original
automaton; transitions leading back from the last to the first level
are called back-edges.  These back-edges play a crucial role as they
are the only colored transitions in the degeneralized SCC and
represent the accepting transitions.

Hence any Büchi accepting cycle in the degeneralization needs to
contain at least one such back-edge, we can therefore focus our
attention on these.  We proceed to check for each back-edge whether we
can embed it in a $(c,b)$-feasible cycle within the degeneralized SCC.
This needs to be done with care and we might need multiple iterations
to ensure that no such cycle exists.
To this end, we start by computing the energy-optimal paths starting at the
destination of the current back-edge (by construction, a state in the
first level) with an initial credit corresponding to its maximal
prefix energy (as found in the first step).
This is done in line~\ref{algBE:firstIter}
and allows us to compute the maximal energy achievable in the destination of the
back-edge when imposing the back-edge as the last transition to be taken.

If this energy is greater than or equal to the prefix energy (in fact it can only be equal, as the prefix
energy is the maximal energy attainable for this state without any additional
constraints), then we can obviously traverse the same
path over and over and have therefore found a $(c,b)$-feasible accepting lasso.

However, the converse is not true. If the computed energy for the destination
of the back-edge is smaller than the prefix energy, we cannot conclude that
no energy feasible cycle embedding the back-edge exists.

\begin{example}
  \label{ex:double_checking}
  Consider the example shown in Figure~\ref{fig:double_checking_orig}.
  Here we have an automaton for which we have to compute
  maximal energy levels in the SCC twice (lines 10-13 in
  Algorithm~\ref{alg:second_part}):
  First we compute the prefix energy from state $0$,
  with $b=30$ and $c=0$. Then we are interested in the only back-edge, leading
  from $(2,2)$ to $(1,1)$ (the states in the degeneralized SCC).
  The state $(1,1)$, the destination of the back-edge, corresponds to state
  $1$ as it is in the first level of the degeneralization, which is rooted in
  the original graph.
  The prefix energy of state $1$ is $30$, while its optimal energy on the cycle,
  after taking the back-edge, is $20$.
  This means that despite it being part of a energy-positive loop the state
  has less energy than after the prefix.
  Hence we cannot conclude that we have found an accepting lasso after the
  first iteration, but need to run the algorithm once more from the state
  $(1,1)$ with a new initial credit of $20$ and $b=30$.
  In this iteration the state $(2,2)$ can once again reach an energy of $30$,
  causing the new energy of state $(1,1)$ to be $20$ once more.
  Now we can finally conclude that a feasible lasso indeed exists.
\end{example}

\begin{figure}[t]
  \begin{subfigure}[t]{0.4\textwidth}
  \centering
  \begin{tikzpicture}[state/.style={shape=rectangle, rounded
        corners, draw}, y=1.4cm, x=1.4cm, node distance=1.3cm and 1.3cm]
    \node[state, initial left] (0) {0};
    \node[state, right=of 0] (1) {1};
    \path (0) edge node {$30$} (1);
    \node[state, right=of 1] (2) {2};
    \path (1) edge[bend left] node {$0$} (2);
    \path (2) edge[bend left] node {$-10$} node[anchor=center]
        {$\color{orange} \bullet$} (1);
    \path (2) edge[loop right] node {$+1$} ();
    \path (2) edge[loop below] node {$-1$} node[anchor=center]
        {$\color{blue} \bullet$}();
  \end{tikzpicture}
  \subcaption{Original WBA}
  \label{fig:double_checking_orig}
  \end{subfigure}
  \hfill
  \begin{subfigure}[t]{0.58\textwidth}
  \centering
  \begin{tikzpicture}[state/.style={shape=rectangle, rounded
        corners, draw}, y=1.4cm, x=1.4cm, node distance=1.0cm and 1.3cm]
    \node[state, initial left] (0) {0};
    \node[state, right=of 0] (11) {1,1};
    \node[state, below=of 11] (12) {1,2};
    \path (0) edge node {$30$} (11);
    \node[state, right=of 11] (21) {2,1};
    \node[state, below=of 21] (22) {2,2};
    \path (11) edge[bend left=10] node {$0$} (21);
    \path (12) edge[bend right=10] node[below] {$0$} (22);
    \path (21) edge[bend left=10] node {$-10$} (11);
    \path (22) edge[bend left=10] node {$-10$} node[anchor=center]
        {$\color{red} \bullet$} (11);
    \path (21) edge[loop right] node {$+1$} ();
    \path (22) edge[loop right] node {$+1$} ();
    \path (22) edge[loop below] node {$-1$} ();
    \path (21) edge[bend left=10] node {$-1$} (22);
  \end{tikzpicture}
  \subcaption{Degeneralizing SCC $\{1, 2\}$ with level $1$ rooted in
    the original WBA. Back-edges colored red.}
  \label{fig:double_checking_degen}
  \end{subfigure}
  \caption{%
    Left: WBA (also used in Example~\ref{ex:double_checking});
    right: degeneralization of one SCC (states named
    \textit{original state}, \textit{level}).}
\end{figure}
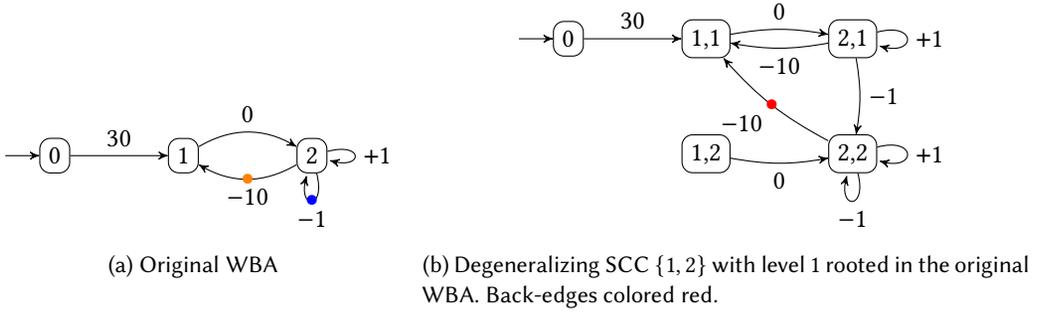

In some cases however, even two iterations do not suffice.
In fact the number of iterations necessary, when simply repeating the
application of the modified Bellman-Ford with an updated initial credit,
can be linear in the value of the weak upper bound $b$ (given enough states exist)
as shown in the next example

\begin{example}
  \label{ex:triple_check}
  Consider the WBA of Figure~\ref{fig:triple_check} which illustrates this problem
  and can be easily extended to necessitate $\mathcal{O}(b)$ iterations.

  The idea is as follows, illustrated for $b=5$:
  We have a single back-edge with a weight of $0$.
  Between the destination and source state of the back-edge (here state $1$ and $6$)
  we create $b-1$ states with positive self-loops (here states $2$ through $5$).
  The prefix for the destination of the back-edge allows it to reach the
  maximal energy of $b$, here the prefix is simply the transition from
  $0$ to $1$ with a weight of $5$.
  The only feasible cycle goes from state $1$ to state $b$ (here state $5$) to
  state $b+1$ (here state $6$), then finally takes the back-edge to complete the
  cycle. On this cycle, the energy attained by state $b+1$ is equal to $1$.

  However, to find this feasible cycle, we first need to discard the cycles
  $1 \to i \to 6 \to 1$ for all $i$ from $2$ to $b-1$.
  During the first iteration, the initial credit is equal to $b=5$.
  The ideal path to the source state of the back-edge passes by state $2$ allowing it
  to reach an energy of $b-1 = 4$. This new initial credit for state $1$ now
  however forbids to take the transition to state $2$.
  The updated optimal path to the source of the back-edge now passes through state $3$ and
  allows it to reach an energy of~$b-2 = 3$.

  This continues in a similar manner for all other states up to $b-1$ after which
  the actual feasible cycle is found.
  This example exposes a flaw in \cite[Algorithm~1]{DBLP:conf/fm/DziadekFS23}
  which only uses two iterations to find feasible cycles and thus would fail to find this one.
\end{example}

Note that the above scenario cycles are hidden because they are not energy optimal
but instead have a (low) constant exit energy,
\ie independently of the entrance energy, the exit energy will always be constant.
In the example above,
state 5 can attain energy $b$ and after applying the exit cost of $-4$,
we reach state 6 with an energy of $1$ independently of the entrance energy at state $1$
(if at least the entrance cost $1$ is available).
In any such scenario, the exit energy is only constant
because the weak upper bound $b$ is attained (or surpassed) along the way.

We would like to avoid having to run a number of iterations which is linear in the weak upper bound.
In order to do so, we propose the following strategy:
Instead of updating the initial credit, we can explicitly search
for cycles which embed a state attaining maximal energy as well as the back-edge.
The idea is to find this cycle by doing the following steps
for every state $s_M$ in the strongly connected component that attained maximal energy:
\begin{itemize}
  \item Run the modified Bellman-Ford starting in $s_M$ with initial credit $b$ to
  compute the maximal energy of the source of the back-edge (line 16).
  \item Propagate this energy along the back-edge to find the energy of the
  destination (line 17)
  \item Run the modified Bellman-Ford starting in the destination of the back-edge
  using the energy computed in the last step as initial credit (line 18).
\end{itemize}

If the energy computed for $s_M$ in the last step is once again $b$, that is, we can
return to this state with maximal energy, then we can conclude that we have found
a feasible cycle.
This optimization is interesting from a practical point of view as the number
of states with maximal energy is typically significantly smaller than $b$.
Moreover it allows to bound the number of iterations necessary by the number of states
rather than the value of $b$ which is necessary to have an overall complexity which
is independent of $b$.

Only if none of the nodes attaining maximal energy can be embedded in a
positive cycle containing the current back-edge we continue with the next
back-edge in the SCC or with the next SCC once all back-edges exhausted.
Finally, we can conclude that no $(c,b)$-feasible Büchi accepting lasso exists
once we have exhausted all (accepting) SCCs.

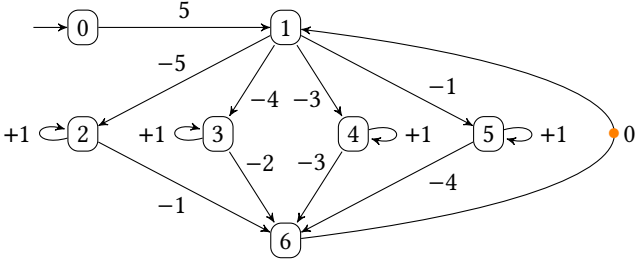
\begin{figure}[t]
  \begin{tikzpicture}[state/.style={shape=rectangle, rounded
        corners, draw}, y=1.4cm, x=1.8cm]
    \node[state, initial left] (0) at (0,0) {0};
    \node[state] (1) at (1.5,0) {1};
    \node[state] (2) at (0,-1) {2};
    \node[state] (3) at (1,-1) {3};
    \node[state] (4) at (2,-1) {4};
    \node[state] (5) at (3,-1) {5};
    \node[state] (6) at (1.5,-2) {6};
    \path (0) edge node {$5$} (1);
    \draw (6) .. controls (4.7,-1.5) and (4.7, -.5) .. node[swap] {$0$}
    node[anchor=center] {$\color{orange} \bullet$} (1);
    \path (2) edge[loop left] node {$+1$} ();
    \path (3) edge[loop left] node {$+1$} ();
    \path (4) edge[loop right] node {$+1$} ();
    \path (5) edge[loop right] node {$+1$} ();
    \path (1) edge[swap] node[xshift=1.5mm] {$-5$} (2);
    \path (1) edge node[xshift=-1.5mm] {$-4$} (3);
    \path (1) edge[swap] node[xshift=1.5mm] {$-3$} (4);
    \path (1) edge[near end] node[xshift=-1.5mm] {$-1$} (5);
    \path (2) edge[swap] node[xshift=1.5mm] {$-1$} (6);
    \path (3) edge[pos=.4] node[xshift=-1.5mm] {$-2$} (6);
    \path (4) edge[swap,pos=.4] node[xshift=1.5mm] {$-3$} (6);
    \path (5) edge[near start] node[xshift=-1.5mm] {$-4$} (6);
  \end{tikzpicture}
  \caption{WBA (of Example~\ref{ex:triple_check}) where two iterations do not suffice}
  \label{fig:triple_check}
\end{figure}

\subsection*{Computing the Optimal Energy}
\label{sec:comp_opt_en}

Our main Algorithm \ref{alg:second_part} allows us to decide the existence
of feasible lassos in the WBA. However one of the key components,
the function \textsc{FindMaxE} has not yet been detailed and we will do so
in this section, detailing how to efficiently find the optimal energy in
weighted graphs allowing positive loops.

The problem is similar (but inverse) to finding shortest paths in weighted
graphs.
This may be done using the well-known Bellman-Ford algorithm
\cite{bellman, ford}, which breaks with an error if it finds negative
loops.
In our inverted problem, we are seeking to maximize energy,
so positive loops are accepted and even desired.
To take into account this particularity, we modify the Bellman-Ford algorithm to
invert the weight handling and to be able to handle positive loops.
The modified Bellman-Ford algorithm is given in Algorithm~\ref{alg:modBF}.

The standard algorithm computes shortest paths by relaxing the distance
approximation until the solution is found.
One round (an iteration of the outer loop over the number of states)
considers all transitions to relax the respective destination node and
the algorithm makes as many rounds as there are nodes.
This ensures that the shortest distance is found as further improvements
can only be caused by negative loops.

Inverting the algorithm is easy: the relaxation is done if the new energy
is higher than the old one; additionally the new energy has to be non-negative
and is bounded from above by the weak upper bound.

The second modification to Bellman-Ford is the handling of positive loops.
This part is more involved, especially if one strives for an efficient algorithm.
We could run Bellman-Ford until it reaches a fixed point, however this can
significantly impact performance as shown in the following example.

\begin{example}
\label{ex:pumping}
Consider the automaton shown in Figure~\ref{fig:pumping} with $b$ being
a multiple of $N$ for simplicity.
Here, in order to attain the maximal energy on a state starting with $0$ initial credit,
$\frac{b}{N}$ iterations of the modified Bellman-Ford are necessary.
This is because the self-loop increases the energy of the state by one,
and the self-loop is considered $N$ times during one iteration.
The energy is however only improved for states which have already been discovered
and the next state (the state to the right of the current state)
can only be reached once the current state has attained $b$.
To reach a fixed point of the energy over the entire graph we need to
reach all states and then ensure that they all reach maximal energy,
which is only achieved after $N\frac{b}{N} = b$ iterations of the modified
Bellman-Ford.
\end{example}

Ideally we would like the upper bound to have no influence on the
runtime.
To this end we introduce the function \textsc{PumpAll},
which sets the energy level of all states on positive loops
detected by the last iteration of Bellman-Ford to the achievable maximum.
This way, instead of needing $\frac{b}{N}$ iterations of Bellman-Ford to
attain the maximal energy, we only need one plus a call to \textsc{PumpAll}.

Before continuing, we make the following observation.
This stage will be called from Algorithm~\ref{alg:second_part}
that recognizes loops necessary to fulfill the Büchi condition.
Here, we only need to check reachability.
Therefore, the only reason to form a loop is to
gain energy, implying that we are only interested in \emph{simple} energy
positive loops, \ie loops where every state appears at most once.
If we set the optimal reachable weight in simple loops, then nested
loops are updated by Bellman-Ford in the usual way afterwards.

To improve the runtime of our algorithm, we exploit that Bellman-Ford
can detect positive loops and handle these loops specifically.  Note
however that contrary to a statement
in~\cite{DBLP:conf/formats/BouyerFLMS08}, we cannot simply set all
energy levels on a positive loop to $b$: in the example of
Figure~\ref{fig:double_checking_orig}, starting in state 2 with an
initial credit of~$10$, the energy level in state 1 will increase with
every round of Bellman-Ford but never above $20=b-10$.

\begin{figure}[tbp]
  \centering
  \begin{tikzpicture}[state/.style={shape=rectangle, rounded
        corners, draw}, y=2cm, x=1.8cm]
    \node[state, initial left] (0) {0};
    \path (0) edge[loop above] node {$+1$} ();
    \node[state, right=of 0] (1) {1};
    \path (0) edge node {$-b$} (1);
    \path (1) edge[loop above] node {$+1$} ();
    \node[state, right=of 1] (2) {2};
    \path (1) edge node {$-b$} (2);
    \path (2) edge[loop above] node {$+1$} ();
    \node[state, right=of 2] (3) {3};
    \path (2) edge node {$-b$} (3);
    \path (3) edge[loop above] node {$+1$} ();
    \node[right=of 3] (4) {\large$\ldots$};
    \path (3) edge node {$-b$} (4);
    \node[state, right=of 4] (5) {$N-1$};
    \path (4) edge node {$-b$} (5);
    \path (5) edge[loop above] node {$+1$} node[anchor=center]
        {$\color{blue} \bullet$}();
  \end{tikzpicture}
  \caption{WBA for Example~\ref{ex:pumping}}
  \label{fig:pumping}
\end{figure}
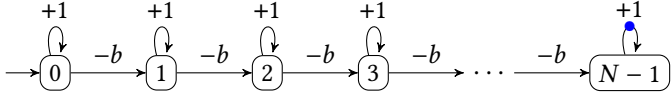

\begin{algorithm}[tbp]
\caption{Modified Bellman-Ford}\label{alg:modBF}
\begin{algorithmic}[1]
\item[]\noindent \hskip-\leftmargin\textbf{Shared Variables:} $E, P$
\Function{modBF}{weighted graph $G$}
  \For{$n \in \{1, \dots, |S|\}$}
    \ForAll{$t=s\xrightarrow{w} s' \in T$}
      \State $e' \gets \min(E(s)+w, b)$
      \If{$E[s'] < e'$ and $e'\geq 0$}
        \State $E[s'] \gets e'$
        \State $P[s'] \gets t$ \Comment{ $P\colon S\rightarrow T$, mapping states to best incoming transition}
        \label{eBF:l_opt_pred}
      \EndIf
    \EndFor
  \EndFor
\EndFunction
\Statex 

\Statex Helper function assigning optimal energy to all states on the
       energy positive loop containing~$s$
\Function{PumpLoop}{weighted graph $G$, state $s$}
  \ForAll{$s' \in$ \Call{Loop}{s}} \Comment{\Call{Loop}{} returns the states on the loop of $s$ ...}
    \State $E[s'] \gets -1$ \Comment{Special value to detect fixed point}
  \EndFor
  \State $E[P[s].src] \gets b$
  \While {$\top$} \Comment{Loops at most twice}
    \ForAll{$s' \in$ \Call{Loop}{s}} \Comment{... in forward order}
      \State $t \gets P[s']$
      \State $e' \gets \min(b, E[t.src]+t.w)$
      \If{$e' = E[t.dst]$}
        \State Mark loop (and suffix) as done
        \State \Return \Comment{fixed point reached}
      \EndIf
      \State $E[t.dst] \gets e'$
    \EndFor
  \EndWhile
\EndFunction
\Statex 

\Statex Helper function, pumping all energy positive loops induced by $P$
\Function{PumpAll}{weighted graph $G$}
  \ForAll{states $s$ that changed their weight}
    \State $t=P[s]$
    \If{$\min(b, E[t.src]+t.w) > E[s]$}
      \State $s' \gets s$ \Comment{$s$ can be either on the loop or in a suffix of one}
      \Repeat \Comment{Go through it backwards to find a state on the loop}
        \State $s'.mark \gets \top$
        \State $s' \gets t.src$
      \Until{$s'$ already marked}
      \State \Call{PumpLoop}{$G, s'$} \Comment{Pump it}
    \EndIf
  \EndFor
\EndFunction
\Statex 

\Statex Function computing the optimal energy for each state
\Function{FindMaxE}{graph $G$, start state $s_0$, initial credit $c$}
  \State Init($s_0$, $c$)
  \Comment{initialize values in $E$ to $-\infty$ and $E(s_0)=c$}
  \While{$not\ fixed point(E)$} \Comment{Iteratively search for loops, then pump them}
    \State \Call{modBF}{$G$}
    \State \Call{PumpAll}{$G$}
  \EndWhile
  \State \Return copyOf($E$)
\EndFunction
\end{algorithmic}
\end{algorithm}

In order to have an algorithm
whose complexity is independent of $b$, we instead compute the fixed
point from above.
We first make the following observation.

\begin{lemma}
  \label{lem:positive_loops_reach_b}
  On (strictly) energy positive loops,
  there exists at least one state on the loop
  that can attain the maximal energy $b$.
\end{lemma}

\begin{proof}
  Since the loop is energy positive,
  we can increase the energy level at any specific node
  by cycling through the loop.
  This can be repeated until a fixed point is reached.
  This fixed point is only reached when at one of the states
  the accumulated weight reaches $b$
  (or would surpass $b$ but is then cut down to $b$).
  As the increase of energy with every loop is a strictly monotone operation,
  a fixed point will be reached.
\end{proof}

If we knew the precise state that attains maximal energy,
we could set its energy to $b$ and follow
the loop once while propagating the energy, causing every
state on the loop to be set to its maximal achievable energy.
However, not knowing which state will effectively attain $b$, we start
with any state on the loop, set its energy to $b$ and propagate the energy
along the loop until a fixed point is reached.
This is the case after
traversing the loop at most twice.
This is done by the function \textsc{PumpLoop}.

\begin{lemma}
  \label{le:pumploop}
  \textup{\textsc{PumpLoop}} calculates the desired fixed point
  after at most two iterations through the loop.
\end{lemma}

\begin{proof}
  In Algorithm~\ref{alg:modBF}, lines 9 and 10 ensure
  that the fixed point check in line 16 does not detect false positives.
  After setting an arbitrary state's energy to $b$,
  the algorithm iterates through the states in the loop in forward order.

  Consider w.l.o.g.\ the positive loop
  $\gamma = s_1\tto{w_1} s_2\tto{w_2}\dotsm \tto{w_{N-1}} s_N$
  with $s_1 = s_N$.
  By Lemma~\ref{lem:positive_loops_reach_b} we know that there exists at least
  one state $s_j$ with $0\le j < N$ whose maximal energy equals $b$.
  Before the first energy propagating traversal of the loop we set the energy of
  $s_1$ to $b$.
  Two cases present themselves.  If $j=0$, then energy is correctly propagated and we
  reach a fixed point after one traversal.
  In the second case, the energy attainable by $s_1$ is strictly smaller than $b$.
  Propagating from this energy level will over-approximate the energies reached
  by the states $s_0$ through $s_{j-1}$ on the loop, but only until state $s_j$ is
  reached which actually attains $b$. As energy is bounded,
  the energy levels of state $s_j$ and its successors $s_{j+1},\dotsc, s_N$ are correctly calculated.
  This means that after traversing the loop
  $s_j\tto{w_j} \dotsm \tto{w_{N-1}} s_{N} \tto{w_1} s_2\tto{w_2}\dotsm \tto{w_{j-1}} s_j$,
  all energy levels on the loop are correctly calculated and this is
  guaranteed to happen before traversing the original loop twice.

  The corresponding fixed point condition is detected by line 16 which will stop the
  iteration.
  Note that we actually need to check for \emph{changes} in the energy level on line
  16, and not whether some state attained energy $b$,
  as we at this point cannot know whether this energy was reached due to
  over-approximation.
\end{proof}

Note that
the pseudocode shown here is a simplification, as our implementation contains
some further optimizations.  Namely, we implement an early exit in
\textsc{modBF} if we detect that a fixed point is reached, and we keep
track of states which have seen an update to their energy, as this
allows to perform certain operations selectively.

\subsection*{Algorithm complexity}

We are now able to conclude our discussion from Section~\ref{sec:wba}
and show that energy Büchi problems for finite WBA are decidable in polynomial time.

\begin{proof}[Proof of Theorem~\ref{th:wba}]
  For our decision procedure,
  the search for strongly connected components can be done in polynomial
  time.
  Our modified Bellman-Ford algorithm also has polynomial complexity.
  It is called once at the beginning of Algorithm~\ref{alg:second_part}
  and then for every back-edge of every strongly connected component,
  it is called several times.
  Its amount depends on the number of energy-maximal states
  (line~\ref{alg:second_part:more_iterations} of Algorithm~\ref{alg:second_part}).
  Given that the number of back-edges is bounded by the number of edges,
  and that the number of energy-maximal states is bounded by the number of states,
  we conclude that our overall algorithm has polynomial complexity.
\end{proof}

\section{Benchmarks}

We employ our running example to build a scalable benchmark case.
For modeling convenience we use
products of WTBAs
as introduced above extended with standard sender/receiver synchronization via channels.
The additional labels $s!$ and $s?$ are used for synchronization.
Edges with $s!$ can always be taken and emit the signal $s$;
edges with $s?$ can only be taken if a signal $s$ is currently emitted.
This modeling allows multiple work modules to start working at the same
time.

\begin{figure}[tbp]
\begin{minipage}[b]{0.4\linewidth}
  \centering
  \begin{tikzpicture}[x=.75cm]
  \node[state with output, initial left, rectangle split,
        rectangle split parts=2, rounded corners] (0) at (0,0) {$x\le
    35$ \nodepart{second} $-10$};
  \node[state with output, rectangle split,
        rectangle split parts=2, rounded corners] (1) at (4,0) {$x\le 55$
    \nodepart{second} $+40$};
  \path (0) edge[out=10, in=170] node[align=center] {$x=35$\\$x\gets 0,s!$} (1);
  \path (1) edge[out=-170, in=-10] node[align=center] {$x=55$\\$x\gets 0,s!$} (0);
  \path (0) edge[loop above] node[above] {$s!$} (0);
  \end{tikzpicture}
  \caption{Base circuit}
  \label{circ:bench_base}

  \vspace*{4ex}

  \begin{tikzpicture}[x=.75cm]
  \node[state with output, initial left, rectangle split,
        rectangle split parts=2, rounded corners] (0) at (0,0) { \nodepart{second} $0$};
  \node[state with output, rectangle split,
        rectangle split parts=2, rounded corners] (1) at (4,0) {$x\le i$
    \nodepart{second} $-10$};
  \path (0) edge[out=10, in=170] node[align=center] {$x\gets 0,s?$} (1);
  \path (1) edge[out=-170, in=-10] node[align=center] {$x=i$} (0);
  \end{tikzpicture}
  \captionof{figure}{Work module $\#i$}
  \label{circ:bench_wm}
\end{minipage}
\quad
\begin{minipage}[b]{0.55\linewidth}
  \captionsetup{type=table}
  \centering
  \begin{tabular}{ c | r | r | r}
    \#mod & \#states & to cpa [s] & sol [s] \\\hline
    1 & 25 & 0.01 & 0.00 \\\hline
    3 & 90 & 0.03 & 0.02 \\\hline
    5 & 293 & 0.06 & 0.24 \\\hline
    7 & 1012 & 0.19 & 3.24 \\\hline
    9 & 3759 & 0.89 & 59.52 \\\hline
    10 & 7377 & 1.87 & 261.38 \\\hline
    11 & 14582 & 4.37 & 1194.81 \\\hline
  \end{tabular}

  \vspace*{3ex}

  \caption{Benchmark results. From left to right:
  Number of work modules, Number of states in cpa,
  time needed to compute cpa, time needed to solve energy Büchi problem.
  Benchmarks done on an ASUS G14, Ryzen 4800H CPU with 16Gb RAM.}
  \label{tab:bench}
\end{minipage}
\end{figure}

As before, we use a base circuit with two states, see Figure~\ref{circ:bench_base}.
Work module $\#i$, see Figure~\ref{circ:bench_wm},
uses $10$ energy units while working
and spends exactly $i$ time units in the work state.
We then combine these models with the specification that time must pass and that
every work module is activated infinitely often.
All the presented instances are schedulable.
Table~\ref{tab:bench} presents the results of our benchmark, showing
that the presented approach scales fairly well.
We note that most of the time for solving the energy Büchi problem (last column)
is spent in our Python implementation of our modified Bellman-Ford algorithm.
In fact the total runtime is (at least for $\text{\#mod}\ge 5$)
directly proportional to the number of times lines 4 to 7 of \textsc{ModBF} in
Algorithm~\ref{alg:modBF} are executed.  Therefore, the implementation could greatly benefit from a
direct integration into Spot and using its
\texttt{C++} engine.

\section{Trace extraction}
\label{se:trace_ext}

Our main Algorithm \ref{alg:second_part} allows us to answer the
question if at least one accepting feasible lasso exists.
However, the algorithm does not provide the lasso itself.
In fact deducing the lasso, which is of great practical use, from the
intermediate results generated by Algorithm \ref{alg:second_part} is a nontrivial task in itself.
This corresponds to the energy Büchi trace problem and will be discussed in this section.

Consider our running example: affirming or refuting the existence
of a feasible schedule for all work modules is of a certain interest.
Extracting the actual trace which can then be used as a control
strategy is however significantly more interesting.

\begin{example}
  \label{ex:trace_difficult}
  Before going into the details of the algorithm, consider the WBA
  given in Figure~\ref{fig:ex.1.tr_ext_base} with $1$ being the initial state,
  the initial credit being set to $0$ and a weak upper bound equal to $100$.

  In this example, all transitions need to be taken infinitely often:
  the transition $2\rightarrow 1$ is needed to satisfy the acceptance condition,
  it can however only be taken if the maximal amount of energy was accumulated;
  the loop $2\rightarrow 4 \rightarrow 2$ is energy positive and allows
  state $2$ to attain maximal energy, however due to its entrance
  cost of $50$ it cannot be taken directly after arriving in $2$ from $1$;
  in order to be able to take the loop $2\rightarrow 4 \rightarrow 2$
  one has first to traverse sufficiently often $2\rightarrow 3 \rightarrow 2$,
  which has no entrance cost, but does not allow the energy in state $2$ to grow
  beyond $50$.

  Ideally we would like to find the (shortest) accepted cycle which in this case is
  $(1 \rightarrow (2 \rightarrow 3 \rightarrow 2)^{50} \rightarrow  (2 \rightarrow 4 \rightarrow 2)^{50} \rightarrow 1 )^\omega$.
  In this work we restrict ourselves to the easier task of finding an abstraction of the
  accepting cycle of the form
  $(1 \rightarrow (2 \rightarrow 3 \rightarrow 2)^+ \rightarrow  (2 \rightarrow 4 \rightarrow 2)^+ \rightarrow 1 )^\omega$
  where, by abuse of the usual notation, $(2 \rightarrow 3 \rightarrow 2)^+$ means that
  the loop $(2 \rightarrow 3 \rightarrow 2)$ is repeated until an energy fixed point
  is reached.

  So why is it difficult to retrieve this trace from the results of Algorithm \ref{alg:second_part}?
  The results are shown in Figure~\ref{fig:ex.1.tr_ext_res}:
  here states are labeled by the maximum energy for the state and the transitions
  leading to the energy optimal predecessor are shown in red.

  This induces multiply difficulties.
  First, the notion of optimal predecessor only allows us to find the
  (energy positive) loop $(2 \rightarrow 4 \rightarrow 2)$ but not
  $(2 \rightarrow 3 \rightarrow 2)$.
  Secondly, the transition $2\rightarrow 1$ which is necessarily taken as it
  corresponds to the back-edge, is not energy optimal.
  Therefore, even when disregarding energy feasibility, we cannot
  hope to find a feasible cycle embedding the back-edge by simply following
  the energy optimal predecessor.
\end{example}

In the rest of the section we will detail an efficient trace extraction
algorithm. We will show how to avoid a complete (re-)exploration
of the graph and why it is necessary to slightly modify the results of
Algorithm \ref{alg:second_part} to achieve this.

\begin{figure}[tbp]
  \begin{subfigure}[t]{0.52\textwidth}
    \centering
    \begin{tikzpicture}[%
        state/.style={shape=rectangle, rounded corners, draw,
          minimum width=.8cm},
        y=1.4cm, x=1.4cm, node distance=1.6cm and 1.6cm]
      \path[use as bounding box] (0,-.6) -- (3.5,1.5);
      \node[state, initial left] (1) at (0,0) {1};
      \node[state, right=of 1] (2) {2};
      \node[state, above=of 2] (3) {3};
      \node[state, right=of 2] (4) {4};
      \path (1) edge[bend left] node[swap] {$0$} (2);
      \path (2) edge[bend left] node {$-100$} node[anchor=center] {$\color{blue}\bullet$} (1);
      \path (2) edge[bend left] node {$+51$} (3);
      \path (3) edge[bend left] node {$-50$} (2);
      \path (2) edge[bend left] node[swap] {$-50$} (4);
      \path (4) edge[bend left] node {$+51$} (2);
    \end{tikzpicture}
    \subcaption{Nontrivial example of trace extraction (for $b=100$)}
    \label{fig:ex.1.tr_ext_base}
  \end{subfigure}
  \hfill
  \begin{subfigure}[t]{0.46\textwidth}
    \centering
    \begin{tikzpicture}[%
        state/.style={shape=rectangle, rounded corners, draw,
          minimum width=.8cm},
        y=1.4cm, x=1.4cm, node distance=1.6cm and 1.6cm]
      \path[use as bounding box] (0,-.6) -- (3.5,1.5);
      \node[state, initial left] (1) at (0,0) {0};
      \node[state, right=of 1] (2) {100};
      \node[state, above=of 2] (3) {100};
      \node[state, right=of 2] (4) {50};
      \path (1) edge[bend left] node[swap] {$0$} (2);
      \path (2) edge[bend left] node {$-100$} node[anchor=center] {$\color{blue}\bullet$} (1);
      \path (2) edge[bend left, draw=red] node {$+51$} (3);
      \path (3) edge[bend left] node {$-50$} (2);
      \path (2) edge[bend left, draw=red] node[swap] {$-50$} (4);
      \path (4) edge[bend left, draw=red] node {$+51$} (2);
    \end{tikzpicture}
    \subcaption{Results obtained from Algorithm \ref{alg:second_part}.}
    \label{fig:ex.1.tr_ext_res}
  \end{subfigure}
  \caption{Example~\ref{ex:trace_difficult}: automaton to show difficulties in trace extraction.}
\end{figure}
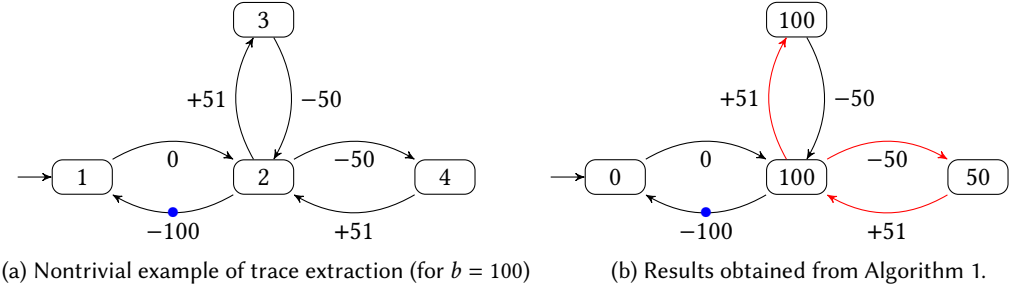

\subsection*{Adapting the solver output}

The standard extension of storing the energy optimal predecessor
(via the corresponding transition) in the Bellman-Ford algorithm,
which corresponds to line~\ref{eBF:l_opt_pred} in Algorithm \ref{alg:modBF}, is not enough.

The problem is that predecessors that have been energy optimal at
some point during the execution of Bellman-Ford are \textit{forgotten}
once better predecessor have been found.
However, these intermediate steps might be crucial to ensure energy feasibility
of the path.
Indeed, to avoid re-exploration of the entire graph we need to store not only
the last energy optimal predecessor, but all predecessors that have been energy
optimal at some point during the execution.  The so-modified procedure can be found
in Algorithm \ref{alg:modBF_trace}.

\begin{algorithm}[tbp]
  \caption{Modified Bellman-Ford All Predecessors}\label{alg:modBF_trace}
  \begin{algorithmic}[1]
  \item[]\noindent \hskip-\leftmargin\textbf{Shared Variables:} $E:Array[int], P:Array[List[int]]$
  \Function{modBF}{weighted graph $G$}
    \For{$n \in \{1, \dots, |S|\}$}
      \ForAll{$t=s\xrightarrow{w} s' \in T$}
        \State $e' \gets \min(E(s)+w, b)$
        \If{$E[s'] < e'$ and $e'\geq 0$}
          \State $E[s'] \gets e'$
          \If{$len(P[s']) < 2$ or $P[s'][-1] \neq t$ or $P[s'][-2] \neq t$}
            \State $P[s'].append(t)$ \Comment{$P\colon S\rightarrow List[T]$}
          \EndIf
        \EndIf
      \EndFor
    \EndFor
  \EndFunction
  \end{algorithmic}
\end{algorithm}

The implications and correctness of the optimization of not unconditionally
storing all predecessors (line 7) will be discussed later on.
For now, simply consider that all \textit{necessary} predecessors are stored.

\subsection*{The extraction algorithm}

With the preliminaries being established, we can detail the actual trace
extraction algorithm.
To better motivate the algorithm we will introduce a notoriously difficult
running example, highlighting most of the encountered problems.

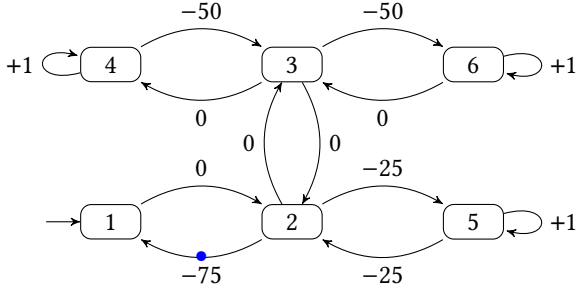
\begin{figure}[tbp]
  \begin{center}
    \begin{tikzpicture}[%
        state/.style={shape=rectangle, rounded corners, draw,
          minimum width=.8cm},
        y=1.4cm, x=1.4cm, node distance=1.6cm and 1.6cm]
       \node[state, initial left] (1) {1};
      \node[state, right=of 1] (2) {2};
      \node[state, above=of 2] (3) {3};
      \node[state, right=of 2] (5) {5};
      \node[state, left=of 3] (4) {4};
      \node[state, right=of 3] (6) {6};
      \path (1) edge[bend left] node {$0$} (2);
      \path (2) edge[bend left] node {$-75$}
      node[anchor=center] {$\color{blue}\bullet$} (1);
      \path (2) edge[bend left] node {$0$} (3);
      \path (3) edge[bend left] node {$0$} (2);
      \path (2) edge[bend left] node {$-25$} (5);
      \path (5) edge[bend left] node {$-25$} (2);
      \path (5) edge[loop right] node {$+1$} ();
      \path (3) edge[bend left] node {$0$} (4);
      \path (4) edge[bend left] node {$-50$} (3);
      \path (4) edge[loop left] node {$+1$} ();
      \path (3) edge[bend left] node {$-50$} (6);
      \path (6) edge[bend left] node {$0$} (3);
      \path (6) edge[loop right] node {$+1$} ();
    \end{tikzpicture}
  \end{center}
  \caption{Trace extraction example. An accepting cycle exists for $b=75$.}
  \label{fig:ex_rabbit}
\end{figure}

\begin{example}
  \label{ex:trace_predecessors_not_unique}
  In Figure~\ref{fig:ex_rabbit},
  every $+1$-self-loop allows to increase energy in that
  state up to the weak upper bound $b=75$.
  The entrance and exist cost (the weights on the incoming and outgoing
  transitions) of the states $4$, $5$ and $6$ put an implicit order
  on which states can be visited.
  State $4$ is the only maximal energy state that can be directly reached from
  state $1$ without any initial credit.
  As the transition from $4$ to $3$ has a cost of $50$, we can only attain an energy
  of $25$ in state $3$, prohibiting a direct passage to state $6$.
  To attain state $6$, which will allow us to eventually traverse the back-edge
  $2 \to 1$, we need to use the self-loop on state $5$ as an additional positive loop.

  Thus, the entrance costs limit the possibilities in which order the self-loop
  states may appear along the trace, \ie first in state $4$, then $5$ and finally $6$.
  Therefore all feasible cycles need to have the form
  \begin{equation*}
    1 \rightsquigarrow (4 \to 4)^{\geq 75} \to 3 \rightsquigarrow 2 \to  (5 \to 5)^{\geq 75} \to 2 \rightsquigarrow 3 \to  (6 \to 6)^{\geq 75} \rightsquigarrow 2 \to 1,
  \end{equation*}
  where $ \tau^{\geq 75} $ means that the loop $\tau$ needs to be
  taken at least $75$ times consecutively and $\rightsquigarrow$ stands for
  any (possibly looping) path that is at least energy-neutral.

  The path segments denoted with $\rightsquigarrow$ allow for instance to take
  the loop $(2 \to 3 \to 2)$ whenever desired, as it is energy neutral.
  Also, in this example it is always possibly to return to positive self-loops
  on states with a lower state number, before continuing.
  For instance, the path $1 \to 2 \to 3 \to (4 \to 4)^{75} \to 3 \to 2 \to (5 \to 5)^{25} \to 2 \to 3 \to (4 \to 4)^{75} \ldots $
  is perfectly feasible as it can be extended to an energy-feasible accepting loop.

  Since the set of feasible cycles is infinite, we cannot consider all of them.
  We need an effective and principled way to investigate them that is correct
  and complete.
  As the Büchi acceptance is ensured by including the back-edge in the cycle,
  the rest of the path's only concern is to ensure energy feasibility.
  This makes it possible to further restrain the structure of paths to be considered
  without losing correctness of the overall algorithm.
\end{example}

\begin{algorithm}[tbp]
  \caption{Trace extraction algorithm}\label{alg:trace_ext}
  \begin{algorithmic}[1]
  \item[]\noindent \hskip-\leftmargin
    \textbf{Shared Variables:}
    \begin{itemize}
    \item $E, E', E'', E_{\rightarrow}, E_{\leftarrow}:Array[int]$ Attainable Energies;
    \item $P, P', P'', P_{\rightarrow}, P_{\leftarrow}:Array[List[transition]]$ Extended predecessor list;
    \item $be: transition$ back-edge to be embedded;
    \item $s_M: int$ Maximal energy state to be embedded;
    \item $G: Graph$, $GS: Graph$ degeneralized SCC
    \end{itemize}
  \Function{TraceExtraction}{}
  \State $src, weight, dst \leftarrow be$
  \If{not $s_M$}
    \Comment{Alg.~\ref{alg:second_part} exited on line 9 or 13, we search directly for a cycle}
    \State $P_{cyc} \gets P''$ if $P''$ else $P'$ \Comment{$P''$ and $E''$ are only set after line 10 in Alg.~\ref{alg:second_part}}
    \State $E_{cyc} \gets E''$ if $E''$ else $E'$
    \State $cyc \gets \textsc{FindPath}(GS, P_{cyc}, dst, src, E_{cyc}[dst], E_{cyc}[src], weight)$
    \State $cyc \gets cyc \cdot be$ \Comment{Concatenate the path with the back-edge to form the cycle}
    \State $\textit{entry} \gets dst$
  \Else
    \Comment{Alg.~\ref{alg:second_part} exited on line 19, we search for a cycle with maximal energy state $s_M$}
    \State $cyc_\rightarrow \gets \textsc{FindPath}(GS, P_\rightarrow, s_M, src, b, E_\rightarrow[src])$
    \State $cyc_\leftarrow \gets \textsc{FindPath}(GS, P_\leftarrow, dst, s_M, min(E_\rightarrow[src] + weight, b), b)$
    \State $cyc \gets cyc_\leftarrow \cdot be \cdot cyc_\rightarrow$ \Comment{Concatenate to form the cycle}
    \State $\textit{entry} \gets s_M$
  \EndIf
  \State $\textit{pref} \gets \textsc{FindPath}(G, P, initialstate, \textit{entry}, c, E[\textit{entry}])$
  \State \Return $\textit{pref} \cdot cyc^+$
  \EndFunction
  \end{algorithmic}
\end{algorithm}

The goal of Algorithm \ref{alg:trace_ext} is to split up the search for lassos into sub-paths that
are easier to handle.
Each of these sub-paths can be found using the corresponding energies and optimal predecessors.
If it contains loops, their only purpose is to accumulate energy, otherwise they represent
the optimal path from a source to a destination.

The shared variables are set in Algorithm~\ref{alg:second_part} where
the extended predecessor lists $P$, $P'$, $P''$, $P_\to$ and $P_\leftarrow$ are
implicitly set at the same time
as their corresponding $E$, $E'$, $E''$, $E_\to$ and $E_\leftarrow$.

The if-part of the algorithm handles the case where we only need to
embed the back-edge in order to find the cycle.
To this end we search for a path from the destination of the back-edge
to the source which is such that the path can be closed to a feasible cycle
if the back-edge has a weight of $w$.
Since this algorithm is only called if the corresponding energy
Büchi problem was feasible, we are guaranteed that such a path exists.

The else-part handles the case where the maximal energy state $s_M$ and the
back-edge need to be embedded into the cycle.
To this end we first search for the optimal path from the maximal energy
state $s_M$ to the source of the back-edge, propagate the energy along
the back-edge and finally search for the optimal path from the destination
of the back-edge back to $s_M$.
As before, we are guaranteed that we can find such paths and that they
form a feasible cycle.

The last step is to search for the prefix of the lasso in a similar manner.
Thereby the search for a lasso is split into finding several easier sub-paths
with additional constraints on the initial credit and desired accumulated energy.
Before detailing the pseudocode of \textsc{FindPath}, let us shortly discuss
its difficulties.

The backward exploration of the extended predecessor list is at its core a combinatorial problem
over all the lists of possible predecessors with repeated elements.
Therefore, standard graph traversal techniques fail as it might be necessary
to traverse the same state multiple times.
For instance in Figure~\ref{fig:ex.1.tr_ext_base}, state $2$ appears in two
different energy positive loops which are both necessary.
In the following we establish several lemmas needed to break the search down to
an efficient algorithm and prove its correctness.

\begin{lemma}
  \label{lemma:no_nested_loops}
  Nested loops are not necessary for energy feasibility.
  That is, every path of the form $s \to (a \to \gamma_1 \to (\tau)^+ \to \gamma_2 \to a)^+ \to d$
  ensuring some maximal energy $m$ in $d$ when starting with initial credit $c$ in $s$ can be
  decomposed into either
  \begin{enumerate}
    \item $s \to a \to \gamma_1 \to (\tau)^+ \to \gamma_2 \to (a \to \gamma_1 \to \gamma_2 \to a)^+ \to d$ or
    \item $s \to a \to \gamma_1 \to (\tau)^+ \to \gamma_2 \to a \to d$
  \end{enumerate}
  with one of them achieving the same $m$ in $d$ when starting with $c$ in $s$.
\end{lemma}

\begin{figure}[tbp]
  \centering
  \begin{tikzpicture}[state/.style={shape=rectangle, rounded
        corners, draw}, y=2cm, x=1.8cm,foo/.style={->,decorate,
        decoration={snake,segment length=1.14mm,amplitude=0.4mm,
          pre length=4pt,post length=4pt,}}]
    \node[state, initial left] (0) {$\vphantom{d}s$};
    \node[state, right=of 0] (1) {$\vphantom{d}a$};
    \path (0) edge node {} (1);
    \node[state, right=of 1] (2) {$d$};
    \path (1) edge node {} (2);
    \node[state, above=of 1] (3) {$\vphantom{d}\hphantom{d}$};
    \path (3) edge[in=295,out=245,loop,foo] node[swap] {$\tau$} ();
    \path (1) edge[bend right=70,foo] node[swap] {$\gamma_1$} (3);
    \path (3) edge[bend right=70,foo] node[swap] {$\gamma_2$} (1);
  \end{tikzpicture}
  \caption{Path illustrating the case of nested loops. Squiggly arrows
  indicate paths instead of transitions.}
  \label{fig:nested_loop}
\end{figure}
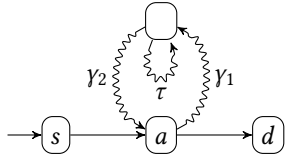

Note that the lemma is without loss of generality:
for multiple nested loops it suffices to reapply the decomposition.

\begin{proof}
  See Figure~\ref{fig:nested_loop} for an illustration.
  The intuition for the different decompositions is as follows:
  In the decomposition (2), the loop $\tau$ (together with the suffix
  $\gamma_2 \to a$) is the energy optimal predecessor for $d$.
  Therefore it is not necessary to take the ``outer'' loop
  $(a \to \gamma_1 \to \tau \to a)$ repeatedly, but $\gamma_1$ and $\gamma_2$ appear
  a single time on the path from $s$ to $d$.

  The decomposition (1) corresponds to the outer loop $a \to \gamma_1 \to \gamma_2 \to a$
  being the energy optimal predecessor for state $d$.
  In this case, $\tau$ may only appear in the prefix of this loop and it might
  be necessary to take it repeatedly to gather enough energy before being able
  to take the optimal loop.
  For instance, imagine $\gamma_1$ to be energy neutral
  and $\gamma_2$ to consist of the sub-paths $\gamma_2'$ and $\gamma_2''$
  ($\gamma_2 = \gamma_2' \to \gamma_2''$).
  If $\gamma_2'$ is significantly energy negative and $\gamma_2''$ is
  significantly energy positive, then it might be necessary to loop on
  $\tau$ in order to be able to complete the outer loop.
  Once we have gathered enough energy to traverse the outer loop once,
  we can always take it again, as it is energy positive overall.
\end{proof}

This allows us to impose an additional structure on the lassos to be considered:
they all have to be of the form
\begin{equation*}
  \gamma_{p,0} {\tau_{p,0}}^+ \gamma_{p,1} {\tau_{p,1}}^+ \cdots \gamma_{p,k} {\tau_{p,k}}^+ .
  (\gamma_{c,0} {\tau_{c,0}}^+ \gamma_{c,1} {\tau_{c,0}}^+ \cdots \gamma_{c,l} {\tau_{c,l}}^+)^\omega.
\end{equation*}
With $k \ge 0$ and $l>0$, all sub-paths denoted $\gamma$ correspond to loop-free segments, all sub-paths denoted with $\tau$ correspond to simple, energy positive, loops.
All sub-paths denoted by $\gamma$ may be empty,
but the loops $\tau$ may not be empty and must be taken at least once.
The subscript $p$ (respectively $c$) indicates that the sub-path belongs to
the prefix (respectively cycle) of the lasso.

Note that we are not interested in determining exactly how often the loops
$\tau$ have to be taken, but simply assume that they are repeated until
an energy fixed point is reached and have to be taken at least once, as denoted by $\tau^+$.

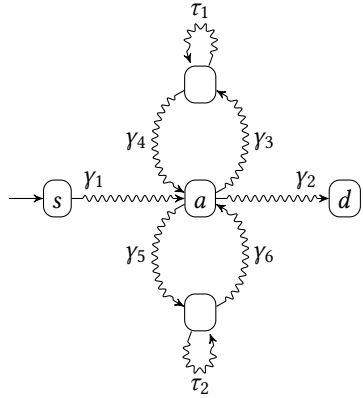
\begin{figure}[tbp]
  \centering
  \begin{tikzpicture}[state/.style={shape=rectangle, rounded
        corners, draw}, y=2cm, x=1.8cm,foo/.style={->,decorate,
        decoration={snake,segment length=1.14mm,amplitude=0.4mm,
          pre length=4pt,post length=4pt,}}]
    \node[state, initial left] (0) {$\vphantom{d}s$};
    \node[state, right=1.5cm of 0] (1) {$\vphantom{d}a$};
    \path (0) edge[foo] node[pos=0.2] {$\gamma_1$} (1);
    \node[state, right=1.5cm of 1] (2) {$d$};
    \path (1) edge[foo] node[pos=0.8] {$\gamma_2$} (2);
    \node[state, above=of 1] (3) {$\vphantom{d}\hphantom{d}$};
    \node[state, below=of 1] (4) {$\vphantom{d}\hphantom{d}$};
    \path (3) edge[in=115,out=65,loop,foo] node[swap] {$\tau_1$} ();
    \path (1) edge[bend right=70,foo] node[swap] {$\gamma_3$} (3);
    \path (3) edge[bend right=70,foo] node[swap] {$\gamma_4$} (1);
    \path (4) edge[in=295,out=245,loop,foo] node[swap] {$\tau_2$} ();
    \path (1) edge[bend right=70,foo] node[swap] {$\gamma_5$} (4);
    \path (4) edge[bend right=70,foo] node[swap] {$\gamma_6$} (1);
  \end{tikzpicture}
  \caption{Non-reappearance of positive loops.}
  \label{fig:reappearance}
\end{figure}

\begin{lemma}
  \label{lem:non_reapp_loop}
  Given two (strictly) energy positive loops $\tau_1$ and $\tau_2$, as shown in Figure~\ref{fig:reappearance},
  some initial credit $c$, an initial state $s$ and a
  weak upper bound $b$,
  let $m$ be the maximal energy attainable in $d$ by one of the following paths:
  \begin{enumerate}
    \item $s \to \gamma_1 \to a \to \gamma_2 \to d$,
    \item $s \to \gamma_1 \to a \to \gamma_3 \to {\tau_1}^+ \to \gamma_4 \to a \to \gamma_2 \to d$,
    \item $s \to \gamma_1 \to a \to \gamma_3 \to {\tau_1}^+ \to \gamma_4 \to a \to \gamma_5 \to {\tau_2}^+ \to \gamma_6 \to a \to \gamma_2 \to d$,
    \item $s \to \gamma_1 \to a \to \gamma_5 \to {\tau_2}^+ \to \gamma_6 \to a \to \gamma_2 \to d$,
    \item $s \to \gamma_1 \to a \to \gamma_5 \to {\tau_2}^+ \to \gamma_6 \to a \to \gamma_3 \to {\tau_1}^+ \to \gamma_4 \to a \to \gamma_2 \to d$.
  \end{enumerate}
  If $\tau_1$ and $\tau_2$ are the only (strictly) energy positive loops in the graph,
  then $m$ is the best energy attainable in $d$ when starting in $s$ with $c$
  among all paths that can be constructed using the sub-paths
  $\gamma_1,\dots, \gamma_6, \tau_1, \tau_2$
  (elements can be repeated).
\end{lemma}

Again, the lemma is without loss of generality:
if more than two (strictly) positive loops exist,
it suffices to reapply the same reasoning multiple times.

\begin{proof}
  From Lemma~\ref{lemma:no_nested_loops} we already know that we do not need to
  consider nested loops such as
  $(a \to \gamma_3 \to {\tau_1}^+ \to \gamma_4 \to \gamma_5 \to {\tau_2}^+ \to \gamma_6 \to a)$.
  This, together with the fact that following the (last / final) optimal
  predecessor leads either back to the initial state or a strictly positive
  loop leaves us with the 5 different cases mentioned above:

  Case (1): The direct path $s \to \gamma_1 \to a \to \gamma_2 \to d$ is optimal,
  passing by $\tau_1$ or $\tau_2$ would not increase the energy in $d$,
  therefore neither $\tau_1$ nor $\tau_2$ appear multiple times.

  Case (2): $\tau_1$ is the optimal predecessor of $d$ and the path
  $s \to \gamma_1 \to a \to \gamma_3 \to {\tau_1}^+ \to \gamma_4 \to a \to \gamma_2 \to d$
  is $(c,b)$-feasible. Passing by $\tau_2$ would not increase the energy
  in $d$ and neither would passing by $\tau_2$ and then $\tau_1$ (it would
  lead to the same energy),
  therefore neither $\tau_1$ nor $\tau_2$ appear multiple times.

  Case (3): As for case (1) $\tau_1$ is the optimal predecessor of $d$
  however we need to gather energy in $\tau_2$ to be able to take it for the
  first time. Once we have enough energy for $\tau_1$ there is no reason to
  return to $\tau_2$ as it is not optimal,
  therefore neither $\tau_1$ nor $\tau_2$ appear multiple times.

  Case (4) and (5) are the symmetric cases for case (2) and (3) only with
  the roles of $\tau_1$ and $\tau_2$ reversed.
\end{proof}

From this follows that we do not need to \textit{revisit} positive loops
along a path in order to obtain the optimal energy in the destination state.
We can therefore restrict the form of the feasible lassos even further, to
\begin{equation}
  \label{eq:path_segments}
  \gamma_{p,0} {\tau_{p,0}}^+ \gamma_{p,1} {\tau_{p,1}}^+ \cdots \gamma_{p,k} {\tau_{p,k}}^+ .
  (\gamma_{c,0} {\tau_{c,0}}^+ \gamma_{c,1} {\tau_{c,0}}^+ \cdots \gamma_{c,l} {\tau_{c,l}}^+)^\omega
\end{equation}
with $\tau_{p, i}\ne \tau_{p, j}$ and $\tau_{c, i}\ne \tau_{c, j}$ for $i\ne j$.

\subsection*{The \textsc{FindPath} algorithm}

The general idea of our \textsc{FindPath} algorithm is to follow the energy
optimal predecessors to find a path from the source to the destination state,
then evaluate whether it is feasible by forward propagating the initial credit
along the path and comparing it against the minimally desired energy at the destination.
However, as it was shown in the initial example, we cannot simply follow
the last predecessor and expect to find a feasible trace.
Nor can we examine all possible paths that can be constructed using the list
of all predecessors that have been optimal at some point, as due to loops there
can be infinitely many of these.
We therefore need to combine the idea of searching only for practical traces
with Lemmas \ref{lemma:no_nested_loops} and \ref{lem:non_reapp_loop}
and the following notion of chronological coherence
in order to construct an efficient algorithm
to find trace candidates.

We say that a run $\rho$ (possibly containing loops)
$\rho = s_0 \to s_1\to \cdots \to a \to (b \to \cdots \to c \to b)^+ \to \cdots$
is \emph{chronologically coherent} with the extended predecessors $P$ if two conditions hold.
First, for all states $s'$ that only appear once with $s \to s'$,
$s$ must be in $P[s']$.
Secondly, for all states $b$ that appear multiple times on $\rho$, we
first create a list $Pa[b]$ of \emph{actual predecessors}.
To this end we traverse $\rho$, and each time we encounter the state $b$
with some $x \to b$, we append $x$ to $Pa[b]$.
Now $\rho$ is chronologically coherent with $P$ if an index array
$Idx$ exists that is monotone and satisfies that for all $i < len(Pa[b])$,
$Pa[b]_i = P[b]_{Idx[i]}$.
Hence the extended predecessor list $P$ contains all predecessors of any $b$,
in the correct order, and consecutive repetitions may be reduced to a single occurence.

\begin{lemma}
  Given
  an initial state $s$ and $c,b\in\Nat$,
  let $E$ and $P$ be the attainable energies and the extended predecessor list
  resulting from the call to $\textup{\textsc{FindMaxE}}(G, s, c)$.
  Then for every state $s'$ there exists a run
  $s\to \cdots \to s'$ that attains the maximal achievable energy $E[s']$ and which
  is chronologically coherent with $P$.
\end{lemma}

\begin{proof}
  Since we traverse every loop on a run until an energy fixed point is reached,
  storing the same predecessor at most twice consecutively is sufficient.
  One entry ensures that a positive loop can be found, the second entry ensures
  that we can retake a loop partially in case the entry-point and the
  exit point of the loop do not coincide.

  Otherwise chronological coherence is ensured by construction as we store
  all predecessors.
\end{proof}

\begin{algorithm}[tbp]
  \caption{\textsc{FindPath} algorithm}\label{alg:path_ext}
  \begin{algorithmic}[1]
  \item[]\noindent \hskip-\leftmargin
    \textbf{Shared Variables:}
    \begin{itemize}
    \item $P:Array[List[transition]]$ Extended predecessor list;
    \item $gSrc: int$ and $gDst: int$ Initial and final state of the trace;
    \item $cSrc: int$ Initial energy at $gSrc$;
    \item $eDst: int$ Minimal energy to attain at $gDst$;
    \item $ext: int\cup None$ Extra cost for loop completion; $None$ if simple path
    \end{itemize}
  \Function{FindPath}{$G$, $P$, $gSrc$, $gDst$, $cSrc$, $eDst$, $ext=None$}
    \LComment{Search for a trace starting in $gSrc$ with $cSrc$ energy to $gDst$ with at least $eDst$ energy}
    \LComment{$ext$ serves as an indicator if we search for an implicit cycle.}
    \State $ci \gets$ [$\textsc{len}(P[s])$ for $s$ in $range(G.numStates())$]
    \State \Return $\Call{BackwardsSearch}{ci, []}$ \Comment{Returned list is empty if no traces exists}
  \EndFunction
  \Statex 
  \Function{BackwardsSearch}{$ci:Array[int]$ current index, $p: List[transition]$ current trace}
    \State $vCurr \gets p.\textsc{front}().src$ if not $p.\textsc{empty}()$ else $gDst$
    \If{$vCurr = gSrc$} \Comment{We found the target}
      \If{$\Call{ForwardExp}{p}$} \Return $p$ \Comment{A feasible trace was found}
      \EndIf
    \EndIf
    \ForAll{$i$ from $ci[vCurr]-1$ to $0$}
      \Comment{The index array $ci$ ensures chronological coherence}
      \State $ci' \gets ci.\textsc{copy}()$
      \State $p' \gets p.\textsc{copy}()$
      \State $ci'[vCurr] \gets i$
      \State $p'.\textsc{pushFront}(P[vCurr][i])$ \Comment{Add the transition to the start of the trace}
      \State $tr \gets \Call{BackwardsSearch}{ci', p'}$
      \If{$tr \neq []$} \Return $p'$
      \EndIf
    \EndFor
    \State \Return $[]$ \Comment{All options exhausted; No feasible trace.}
  \EndFunction
  \Statex 
  \Function{ForwardExp}{$p$ candidate trace}
    \LComment{The extra cost $ext$ is set to None if we do not search for a cycle}
    \State $tc \gets \Call{CompressTrace}{p}$ \Comment{Decompose trace into loops and prefixes}
    \State $e \gets cSrc$ \Comment{Initial energy is defaulted to initial credit}
    \State $eTarget \gets eDst$
    \ForAll{$\_$ in $range(1+(ext \text{ is } None))$} \Comment{Loop twice if $ext$ is given}
      \ForAll{$\textit{pref}, cyc$ in $tc$}
        \State $succ, e \gets \Call{propAlong}{e, \textit{pref}}$ \Comment{Propagate energy along prefix}
        \If{not $succ$} \Return False \Comment{Prefix was not energy feasible}
        \EndIf
        \State $succ, e \gets \Call{tryPumpLoop}{e, cyc}$
        \Comment{Detects whether loop is energy positive \emph{and} feasible}
        \If{not $succ$} \Return False \Comment{Loop was not energy feasible or not energy positive}
        \EndIf
      \EndFor
      \If{$e \ge eTarget$} \Return True \Comment{Check whether enough energy was accumulated}
      \EndIf
      \State $e \gets min(e+ext, b)$
      \Comment{Close the implicit cycle}
      \State $eTarget \gets e$
    \EndFor
    \State \Return False
  \EndFunction
\end{algorithmic}
\end{algorithm}

\subsection*{Ensuring feasibility}

Due to the asymmetry of the weak upper bound we cannot simply
compute the correct energy levels via back propagation, even when following
predecessors that have been energy optimal at some point.
We need to use an alternation between backward search and forward exploration.
We first construct a path, which might contain loops, from the destination state
to the source state by following one of the optimal predecessors at each step.
Once such a path is found, we need to check its energy feasibility
using forward exploration (along the path).

The forward exploration
can be done in linear time
in the size of the WBA.
The backward search is more complicated and is only guaranteed to terminate
thanks to the lemmas above.
The $\textsc{FindPath}$ function in Algorithm~\ref{alg:path_ext} starts a backward search over all
chronologically coherent traces.
It switches to a forward exploration once a candidate is found, in order to
evaluate its energy feasibility.
In Algorithm~\ref{alg:path_ext}, the parameters given to $\textsc{FindPath}$
are afterwards shared between the called functions.

\begin{algorithm}[tbp]
  \caption{\textsc{FindPath} algorithm -- Helper functions}
  \label{alg:trace_ext_help}
  \begin{algorithmic}[1]
  \item[]\noindent \hskip-\leftmargin
    \textbf{Shared Variables:}
    \begin{itemize}
    \item $P:Array[List[transition]]$ Extended predecessor list;
    \item $be: transition$ Back-edge of interest;
    \item $c: int$ Initial energy at $be.src$;
    \item $G:Graph$
    \end{itemize}
  \Statex Transforms a trace given as a list of transitions into a list of path segments
  \Function{TraceCompression}{$p: List[transition]$ the trace to compress}
    \State $idx \gets 0$
    \State $tc \gets []$ \Comment{List containing the path segments}
    \While{$idx < \textsc{len}(p)$} \Comment{Make sure to treat the entire list}
    \State $subtrace \gets []$
    \State $srcIdx \gets Dict()$ \Comment{Look-up mapping src vertex to index in subtrace}
    \While{$idx < \textsc{len}(p)$} \Comment{Search for the next path segment}
      \State $subtrace.\textsc{pushBack}(p[idx])$ \Comment{Add the transition to the end of the trace}
      \State $idx \gets idx + 1$
      \If{$subtrace.\textsc{back}().dst$ in $srcIdx.keys()$} \Comment{We found a loop}
        \State $cutIdx \gets srcIdx[subTrace.\textsc{back}().dst]$
        \State $tc.\textsc{pushBack}(\textsc{pathSegment}(subtrace[:cutIdx], subtrace[cutIdx:]))$
        \State \Call{break}{} \Comment{Advance to next segment}
      \EndIf
      \State $srcIdx[subtrace.\textsc{back}().src] \gets \textsc{len}(subtrace) - 1$
    \EndWhile
    \EndWhile
    \If{$subtrace$} \Comment{add remaining subtraces to $tc$}
      \State $tc.\textsc{pushBack}(\textsc{pathSegment}(subtrace, []))$
    \EndIf
    \State \Return $tc$
  \EndFunction
  \Statex 
  \Statex Propagate Energy along a path
  \Function{PropAlong}{$e:int$ the energy before, $p:List[transitions]$: the path}
    \ForAll{$(s, weight, t)$ in $p$}
      \State $e' \gets min(b, e+weight)$
      \If{$e' < 0'$} \Return False, $0$
      \EndIf
      \State $e \gets e'$
    \EndFor
    \State \Return True, $e$
  \EndFunction
  \Statex 
  \Statex Pump a possibly positive loop
  \Function{tryPumpLoop}{$e:int$ the energy before, $p:List[transitions]$: a cycle}
    \If{not $p$}
      \State \Return True, $e$
    \EndIf
    \State $eInit \gets e$
    \State $succ, e \gets \Call{PropAlong}{e, p}$ \Comment{Take the loop once to determine if positive and necessary}
    \If{(not $succ$) or $e \le eInit$} \Return False, $0$
    \EndIf
    \State $e \gets b$ \Comment{Set energy to the upper bound and correct via propagation}
    \State $\_, e \gets \Call{PropAlong}{e, p}$
    \State $\_, e \gets \Call{PropAlong}{e, p}$
    \State \Return True, $e$
  \EndFunction
  \end{algorithmic}
\end{algorithm}

Algorithm \ref{alg:trace_ext_help} shows some helper functions
in order to do the forward exploration efficiently.
\textsc{TraceCompression} turns a list of transitions into a list of path segments
by collecting and collapsing sequences of concatenable transitions.
A path segment is a pair of a linear trace and a loop part.
A list of path segments build a trace in the sense of Equation~\eqref{eq:path_segments}.
\textsc{PropAlong} is a simple procedure which propagates energy along a path
(or returns \textit{False} if energy drops below $0$).
\textsc{tryPumpLoop} tries to increase energy along a loop:
it returns \textit{False} if the loop is infeasible or energy non-positive;
otherwise it computes the maximal energy fixed point from above.

\begin{example}
  We continue the preceding Example~\ref{ex:trace_predecessors_not_unique} with Figure~\ref{fig:ex_rabbit}.
  Invoking \textsc{FindMaxE} on it with initial credit $0$ and a weak upper
  bound $75$ will return 
  $[1,3,5,3]$ as extended predecessor list of state $2$ (the most interesting state
  in this example).

  This predecessor list is chronologically coherent with the trace
  \begin{align*}
    \tau = \big(1 &\to 2\to 3\to (4\to 4)^{75}\\
                  &\to 3 \to 2 \to (5 \to 5)^{75}\\
                  &\to 2 \to 3 \to (6 \to 6)^{75} \to 3 \to 2 \to 1\big)^+
  \end{align*}

  The trace found by our algorithm is 
  \begin{align*}
    \tau = \big(1 &\to 2\to 3\to (4\to 4)^+\\
                  &\to 3 \to 2 \to (5 \to 5)^+\\
                  &\to 2 \to 3 \to (6 \to 6)^+ \to 3 \to 2 \to 1\big)^+
  \end{align*}
  which is the expected result given that we do not seek to identify
  how often loops have to be taken but always assume retaking
  them until an energy fixed point is reached.

  Note that the extended predecessor list of a state can depend on the number of states
  in the automaton. For instance, if Figure~\ref{fig:ex_rabbit} would only show a part of the automaton
  and there were $50$ additional states, then the extended predescessor list of $2$ would become
  $[1,3,3,5,5,3,3]$. This is due to the additional iterations in the modified Bellman-Ford algorithm,
  allowing the states $3$ and $5$ to reach higher energy levels which will in turn propagate to state $2$.
\end{example}

\section{Parity Condition}
\label{sec:parity_cond}

In this final section we show how to adapt our solution to Parity automata.
Let $(\mathcal{M}, S, s_0, T)$ be a WBA,
but now
$p: \mathcal{M}\to \Nat$ is a function assigning non-negative integers
(\ie~priorities) to colors.
An infinite run $\rho = s_1\to_{M_1} s_2\to_{M_2}\dotsm$ is
\emph{Parity accepted} if the maximal priority seen infinitely
often along $\rho$ is even, that is, if
$\max\{p(m) \mid m \in \text{Inf}((M_i)_{i\geq 1})\}$ is even.
As the number of colors is finite, the maximum always exists.

An example with priorities up to 4 is shown symbolically in Figure~\ref{fig:parity}
with accepting regions represented as white zones, whereas rejecting ones
are colored gray.
A run $\rho$ is accepted if either priority 4 occurs infinitely often, or
priority 3 occurs finitely often and priority 2 occurs infinitely often, or
priority 1 occurs finitely often and priority 0 occurs infinitely often.

\begin{figure}[t]
  \centering
  \begin{tikzpicture}
    \node[right,draw,ellipse,minimum height=8em,minimum width=22em] (4) {};
    \node[right,draw,ellipse,minimum height=7em,minimum width=18em,fill=gray!50] (3) {};
    \node[right,draw,ellipse,minimum height=6em,minimum width=14em,fill=white] (2) {};
    \node[right,draw,ellipse,minimum height=5em,minimum width=10em,fill=gray!50] (1) {};
    \node[right,draw,ellipse,minimum height=3em,minimum width=6em,fill=white] (0) {Inf 0};
    \path (0.east) node[right] {Fin 1};
    \path (1.east) node[right] {Inf 2};
    \path (2.east) node[right] {Fin 3};
    \path (3.east) node[right] {Inf 4};
  \end{tikzpicture}
  \caption{Parity condition for priorities up to 4:\\
  $\text{Inf}(4) \mid (\text{Fin}(3) \,\&\, (\text{Inf}(2) \mid (\text{Fin}(1) \,\&\, \text{Inf}(0))))$}
  \label{fig:parity}
\end{figure}
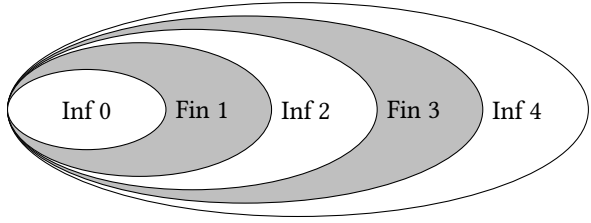

In order to solve energy Parity problems,
we transform them to successive energy Büchi problems.
We calculate prefix energies on the original automaton, as we did in the Büchi case.
For every SCC found by Couvreur's algorithm that contains at least one
transition of even priority (in SCCs containing only odd priorities, all cycles are rejected), we do the following:
First we create a new automaton that is a copy of the current SCC.
Then we give it to an algorithm performing the following steps.
\begin{itemize}
  \item If the automaton is empty, \ie it has no transitions, answer negatively.
  \item If the highest priority is even:
  \begin{itemize}
    \item create a copy of the current automaton;
    \item set the acceptance condition to Büchi;
    \item recolor the automaton the following way: all transitions with highest priority become Büchi accepted, all others are uncolored;
    \item solve the energy Büchi problem for this automaton.
    If the result is positive, we answer positively as well. If not, then we remove all
    transition with the highest priority from the automaton and perform a recursive call.
  \end{itemize}
  \item If the highest priority is odd:
  \begin{itemize}
    \item remove all transitions with the highest priority from the automaton;
    \item perform a recursive call.
  \end{itemize}
\end{itemize}

This shows that the energy Parity problem can be reduced into several
energy Büchi problems.
Moreover, this allows us to use the same algorithm for trace extraction in
the parity case.

\section{Conclusion}

We have shown how to efficiently solve energy Büchi problems, both in
finite weighted (transition-based generalized) Büchi automata and in
one-clock weighted timed Büchi automata, as well as how to efficiently extract
the actual trace from the intermediate results.
We have also extended our results to Parity conditions.

We have implemented all our algorithms in a tool based on TChecker and Spot.  Solving the latter problem is done by using the corner-point abstraction to
translate the weighted timed Büchi automaton to a finite weighted
Büchi automaton; the former problem is handled by interleaving a
modified version of the Bellman-Ford algorithm with Couvreur's
algorithm.

Our tool is able to handle some interesting examples, but the
restriction to one-clock weighted timed Büchi automata without weights
on edges does impose some constraints on modeling.
We believe that trying to lift the one-clock restriction is unrealistic;
but weighted edges (without Büchi conditions) have been treated in
\cite{DBLP:conf/hybrid/BouyerFLM10}, and we suspect that their approach
should also be feasible here.
(See \cite{DBLP:journals/lmcs/CacheraFL19} for related work.)

As a last remark,
it is known that multiple clocks, multiple weight dimensions, and even
turning the weak upper bound into a strict one which may not be
exceeded, rapidly leads to undecidability results, see
\cite{DBLP:conf/formats/BouyerFLMS08, DBLP:journals/pe/BouyerLM14,
  DBLP:conf/ictac/FahrenbergJLS11, DBLP:conf/lata/Quaas11}, and we are
wondering whether some of these may be sharpened when using Büchi
or Parity conditions.

\begin{acks}
  We are grateful to Rania Saadi for her help
  in implementing and testing our trace extraction algorithm
  and the fruitful discussions during her internship.
  The first author
  Sven Dziadek is supported by the \grantsponsor{ANR}{ANR}{https://anr.fr/} project EQUUS,
  grant number \grantnum{ANR}{ANR-19-CE48-0019}, and
  funded by the \grantsponsor{DFG}{Deutsche Forschungsgemeinschaft}{https://www.dfg.de/}
  (DFG, German Research Foundation), grant number \grantnum{DFG}{431183758}.
\end{acks}

\bibliographystyle{plain}
\bibliography{mybib}

\end{document}